\documentclass[twoside,10pt]{article}

\usepackage{arxiv_notation}
\usepackage{xcolor}
\usepackage{hyperref}
\hypersetup{colorlinks=true,linkcolor=violet,citecolor=violet,urlcolor=violet}
\usepackage{jmlr2e}
\usepackage{graphicx}
\usepackage{amsmath}
\usepackage{cleveref}
\usepackage{mathtools}
\usepackage{xspace}
\usepackage{multicol}
\usepackage{multirow}
\usepackage{booktabs}
\usepackage{caption}
\usepackage{subcaption}
\usepackage{algorithm}
\usepackage{algorithmic}

\def \yijs {{\mathbf{y}_{ij}^s}}

\def \ip {i^\prime}
\def \jp {j^\prime}

\def \tp  {t^\prime}
\def \ijp {i^\prime, j^\prime}
\def \sub {\scriptscriptstyle}
\def \btheta {\boldsymbol{\theta}}
\def \bTheta {\boldsymbol{\Theta}}

\def \bSigma {\mathbf{\Sigma}}

\newenvironment{hproof}{%
  \proof}{\endproof}
\newcommand{\la}{\langle}
\newcommand{\ra}{\rangle}
\newcommand{\bthetaj}{\btheta^{(j)}}
\newcommand{\bthetaz}{\btheta^{(0)}}
\newcommand{\tbthetaj}{\tilde \btheta^{(j)}}
\newcommand{\tbthetaz}{\tilde \btheta^{(0)}}

\jmlrheading{1}{2000}{1-48}{4/00}{10/00}{Yiling Jia, Hongning Wang}

\ShortHeadings{Learning Neural Ranking Models Online from Implicit User Feedback}{Jia and Wang}
\firstpageno{1}

\begin{document}

\title{Learning Neural Ranking Models Online from Implicit User Feedback}

\author{\name Yiling Jia \email yj9xs@virginia.edu \\
	\addr Department of Computer Science\\
	University of Virginia\\
	Charlottesville, VA 22903, USA
	\AND
	\name Hongning Wang \email hw5x@virginia.edu\\
	\addr Department of Computer Science\\
	University of Virginia\\
	Charlottesville, VA 22903, USA
}

\maketitle

\begin{abstract}
Existing online learning to rank (OL2R) solutions are limited to linear models, which are incompetent to capture possible non-linear relations between queries and documents. In this work, to unleash the power of representation learning in OL2R, we propose to directly learn a neural ranking model from users' implicit feedback (e.g., clicks) collected on the fly. We focus on RankNet and LambdaRank, due to their great empirical success and wide adoption in offline settings, and control the notorious explore-exploit trade-off based on the convergence analysis of neural networks using neural tangent kernel. 
Specifically, in each round of result serving, exploration is only performed on document pairs where the predicted rank order between the two documents is uncertain; otherwise, the ranker's predicted order will be followed in result ranking. 
We prove that under standard assumptions our OL2R solution achieves a gap-dependent upper regret bound of $O(\log^2(T))$, in which the regret is defined on the total number of mis-ordered pairs over $T$ rounds. Comparisons against an extensive set of state-of-the-art OL2R baselines on two public learning to rank benchmark datasets demonstrate the effectiveness of the proposed solution.
\end{abstract}

\maketitle

\section{Introduction}


In the past decade, advances in deep neural networks (DNN) have made significant strides in improving offline learning to rank models~\citep{burges2010ranknet,pasumarthi2019tf}, thanks to DNN's strong representation learning power. But quite remarkably, most existing work in online learning to rank (OL2R) still assume a linear scoring function \citep{yue2009interactively,schuth2016multileave,wang2019variance}.
Compared with linear ranking models, nonlinear models induce a more general hypothesis space, which provides a system more flexibility and capacity in modeling complex relationships between a document’s ranking features and its relevance quality.
Such a clear divide between the current OL2R solutions and the successful practices in offline solutions seriously restricts OL2R's real-world impact.

The essence of OL2R is to learn from users' implicit feedback on the presented rankings, which suffers from the explore-exploit dilemma, as the feedback is known to be noisy and biased \citep{joachims2005accurately,agichtein2006improving,joachims2007evaluating,chapelle2012large}. 
State-of-the-art OL2R approaches employ random exploration to obtain a trade-off, and mainstream OL2R solutions are mostly different variants of dueling bandit gradient descent (DBGD)~\citep{yue2009interactively}. In particular, DBGD and its extensions ~\citep{yue2009interactively, schuth2016multileave, oosterhuis2016probabilistic, schuth2014multileaved} were inherently designed for linear models, where they rely on random perturbations to sample model variants and estimate the gradient for the model update. Given the complexity of a DNN, such a random exploration method can hardly be effective. 
\citet{oosterhuis2018differentiable} proposed PDGD, which samples the next ranked document from a Plackett-Luce model and estimates an unbiased gradient from the inferred pairwise preference. Though PDGD with a neural ranker reported promising empirical results, its theoretical property is still unknown. Most recently, \citet{jia2021pairrank} proposed to learn a pairwise ranker online using a divide-and-conquer strategy. Improved performance against all aforementioned OL2R solutions was reported by the authors. However, this solution is still limited to linear ranking functions in nature. 

Turning a neural ranker online is non-trivial. While deep neural networks can be accurate on learning given user feedback, i.e., exploitation, developing practical methods to balance exploration and exploitation in complex online learning problems remains largely unsolved. In essence, quantifying a neural model's uncertainty on new data points remains challenging. Fortunately, substantial progress has been made to understand the representation learning power of DNNs.
Studies in~\citep{cao2019generalization1, cao2019generalization2, chen2019much, daniely2017sgd, arora2019exact} showed that by using (stochastic) gradient descent, the learned parameters of a DNN are located in a particular regime, and the generalization error bound of the DNN can be characterized by the best function in the corresponding neural tangent kernel space~\citep{jacot2018neural}. In particular, under the framework of the neural tangent kernel, studies in \citep{zhou2020neural, zhang2020neural} proposed that the confidence interval of the learned parameters of a DNN can be constructed based on the random feature mapping defined by the neural network's gradient on the input instances. These efforts prepare us to study neural OL2R. 

In this work, we 
choose RankNet \citep{burges2010ranknet} as our base ranker for OL2R because of its promising empirical performance in offline settings \citep{chapelle2011yahoo}. 
We devise exploration in the pairwise document ranking space and balance exploration and exploitation based on the ranker's confidence about its pairwise estimation. 
In particular, we construct pairwise uncertainty from the tangent features of the neural network \citep{cao2019generalization1, cao2019generalization2}. In each round of result serving, all the estimated pairwise comparisons are categorized into two types, certain pairs and uncertain pairs. 
Documents associated with uncertain pairs are randomly shuffled for exploration, while the order among certain pairs is preserved in the presented ranking for exploitation. 

We rigorously proved that our model's exploration space shrinks exponentially fast as the ranker estimation converges, such that the cumulative regret defined on the number of mis-ordered pairs has a sublinear upper bound. 
As most existing ranking metrics can be reduced to different kinds of pairwise document comparisons \citep{Wang2018Lambdaloss}, 
we also extended our solution to LambdaRank \citep{quoc2007learning} to directly optimize ranking metrics based on users' implicit feedback on the fly. 
To the best of our knowledge, this is the first neural OL2R solution with theoretical guarantees. Our extensive empirical evaluations also demonstrated the strong advantage of our model against a rich set of state-of-the-art OL2R solutions over two public learning to rank benchmark datasets on standard ranking metrics.

\section{Related Work}

\noindent\textbf{Online learning to rank.}
We broadly group existing OL2R solutions into two main categories.
The first type learns the best ranked list for each individual query separately, by modeling users' click and examination behaviors with multi-armed bandit algorithms \citep{radlinski2008learning,kveton2015cascading,zoghi2017online,lattimore2018toprank}. 
Typically, solutions in this category depend on specific click models to decompose relevance estimation on each query-document pair; as a result, exploration is performed on the ranking of individual documents. For example, by assuming users examine documents from top to bottom until reaching the first relevant document, cascading bandit models rank documents based on the upper confidence bound of their estimated relevance \citep{kveton2015cascading, kveton2015combinatorial, li2016contextual}. 
The second type of OL2R solutions leverage ranking features for relevance estimation, and search for the best ranker in the entire model space \citep{yue2009interactively,li2018online,oosterhuis2018differentiable}. The most representative work is Dueling Bandit Gradient Descent (DBGD) \citep{yue2009interactively,schuth2014multileaved}.
To ensure an unbiased gradient estimate, DBGD uniformly explores in the entire model space, which costs high variance and high regret during online ranking and model update. 
Subsequent methods improved DBGD by developing more efficient sampling strategies, such as multiple interleaving and projected gradient, to reduce variance \citep{hofmann2012estimating,zhao2016constructing,oosterhuis2017balancing, wang2018efficient, wang2019variance}. 


However, almost all of the aforementioned OL2R solutions are limited to linear models, which are incompetent to capture any non-linear relations between queries and documents. This shields OL2R away from the successful practices in offline learning to rank models, which are nowadays mostly empowered by deep neural networks \citep{burges2010ranknet,pasumarthi2019tf}. This clear divide has motivated some recent efforts. \citet{oosterhuis2018differentiable} proposed PDGD which samples the next ranked document from a Plackett-Luce model and estimates gradients from the inferred pairwise result preferences. Though PDGD with a neural ranker achieved empirical improvements, there is no theoretical guarantee on its performance. A recent work learns a pairwise logistic regression ranker online and reports the best empirical results on several OL2R benchmarks \citep{jia2021pairrank}. Though non-linearity is obtained via the logistic link function, its expressive power is still limited by the manually crafted ranking features.

\noindent\textbf{Theoretical analysis of neural networks.}
Recently, substantial progress has been made to understand the convergence of deep neural networks \citep{liang2016deep, telgarsky2015representation, telgarsky2016benefits, yarotsky2017error, yarotsky2018optimal, lu2017depth, hanin2017approximating, zou2018stochastic, zou2019improved}. A series of recent studies showed that (stochastic) gradient descent can find global minimal of training loss under moderate assumptions~\citep{liang2016deep, du2018gradient, allen2019convergence, zou2019improved, zou2020gradient}. Besides, \citet{jacot2018neural} proposed the neural tangent kernel (NTK) technique, which describes the change of a DNN during gradient descent based training. This motivates the theoretical study of DNNs with kernel methods. 
Research in~\citep{cao2019generalization1, cao2019generalization2, chen2019much, daniely2017sgd, arora2019exact} showed that by connecting DNN with kernel methods, (stochastic) gradient descent can learn a function that is competitive with the best function in the corresponding neural tangent kernel space.
In particular, under the framework of NTK, some recent work show that the confidence interval of the learned parameters of a DNN can be constructed based on the random feature mapping defined by the neural network's gradient \citep{zhou2020neural, zhang2020neural}. This makes the quantification of a neural model's uncertainty possible, and enables our proposed uncertainty-based exploration for neural OL2R.

\section{Method}
\label{sec:method}
In this section, we present our solution, which trains a neural ranking model with users' implicit feedback online. 
The key idea is to partition the pairwise document ranking space and only explore the pairs where the ranker is currently uncertain while exploiting the predicted rank of document pairs where the ranker is already certain. We rigorously prove a sublinear regret which is defined on the cumulative number of mis-ordered pairs over the course of online result serving. 

\subsection{Problem Setting}
In OL2R, at round $t \in [T]$, the ranker receives a query $q_t$ and its associated $V_t$ candidate documents represented by a set of $d$-dimensional query-document feature vectors: $\mathcal{X}_t = \{\xb_1^t, ..., \xb_{V_t}^t\}$ with $\xb^t_i \in \RR^d$. The ranking $\tau_t = \big(\tau_t(1), ..., \tau_t(V_t)\big) \in \Pi([V_t])$, is generated by the ranker based on its knowledge so far, where $\Pi([V_t])$ represents the set of all permutations of $V_t$ documents and $\tau_t(i)$ is the rank position of document $i$.

The user examines the returned ranked list and provides his/her feedback, i.e., clicks $C_t = \{c_1^t, c_2^t, ..., c_{V_t}^t\}$, where $c_i^t = 1$ if the user clicked on document $i$ at round $t$; otherwise $c_i^t = 0$. Then, the ranker updates itself according to the feedback and precedes the next round. Numerous studies have shown $C_t$ only delivers implicit relevance feedback, and it is subject to various biases and noise, e.g., presentation bias and position bias \citep{joachims2005accurately,agichtein2006improving,joachims2007evaluating}. In particular, it is well-known that non-clicked documents cannot be simply treated as irrelevant. Following the practice in \citep{joachims2005accurately}, we treat clicks as relative preference feedback and assume that clicked documents are preferred over the \emph{examined} but unclicked ones. In addition, we adopt a simple examination assumption: every document that precedes a clicked document and the first subsequent unclicked document are examined. This approach has been widely employed and proven effective in learning to rank \citep{wang2019variance,agichtein2006improving,oosterhuis2018differentiable}. We use $o_t$ to represent the index of the last examined position in the ranked list $\tau_t$ at round $t$. It is worth mentioning that our solution can be easily adapted to other examination models, e.g., position based model \citep{craswell2008experimental}, as we only use the derived result preferences as model input.

As the ranker learns from user feedback while serving, cumulative regret is an important metric for evaluating OL2R. In this work, our goal is to minimize the following regret, which is defined by the number of mis-ordered pairs from the presented ranked list to the ideal one, i.e.,  the Kendall's Tau rank distance,
\small
\begin{equation}
\label{eq_regret_def}
    R_T = \mathbb{E}\left[\sum\nolimits_{t=1}^T r_t\right] = \mathbb{E} \left[\sum\nolimits_{t=1}^T K(\tau_t, \tau_t^*)\right]
\end{equation}
\normalsize
where $K(\tau_t, \tau_t^*)=|\{(i,j):i{<}j,\big(\tau_{t}(i){<}\tau_{t}(j)\wedge \tau^*_{t}(i){>}\tau^*_{t}(j)\big)\vee \big(\tau_{t}(i){>}\tau_{t}(j)\wedge \tau^*_{t}(i){<}\tau^*_{t}(j)\big)\}|$.
\begin{remark}
As shown in \citep{Wang2018Lambdaloss}, most ranking metrics, such as Average Rank Position (ARP) and Normalized Discounted Cumulative Gain (NDCG), can be decomposed into pairwise comparisons; hence, this regret definition connects an OL2R algorithm's online performance with classical rank evaluations. We consider it more informative than ``pointwise'' regret defined in earlier work \citep{lattimore2018toprank,kveton2015cascading}. 
\end{remark}

\subsection{Online Neural Ranking Model Learning}

In order to unleash the power of representation learning of neural models in OL2R, we propose to directly learn a neural ranking model from its interactions with users. We balance the trade-off between exploration and exploitation based on the model's confidence about its predicted pairwise rank order. The high-level idea of the proposed solution is explained in Figure~\ref{fig:model}.

\begin{figure}
    \centering
    \includegraphics[width=\linewidth]{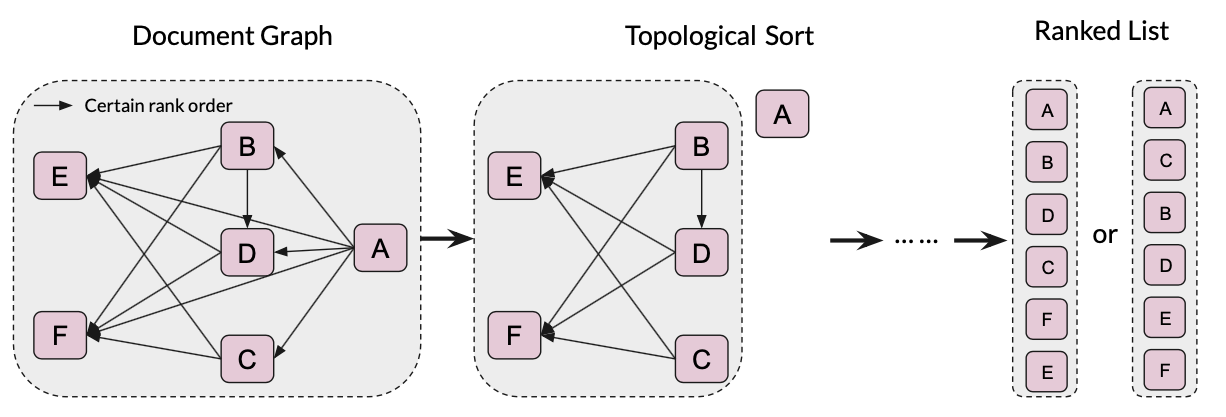}
    \vspace{-5mm}
    \caption{At the current round $t$, the ranker is confident about its rank order estimation between all the pairs expect $(B, C), (C, D), (E, F)$. Hence, in its output ranking, the ranking orders among the certain pairs are preserved, while the uncertain pairs are shuffled.} 
    \vspace{-5mm}
    \label{fig:model}
\end{figure}



\noindent\textbf{Neural Ranking Model.}
We focus on RankNet and LambdaRank because of their promising empirical performance and wide adoption in offline settings \citep{burges2010ranknet}. In the following sections, we will focus our discussion on RankNet to explain the key components of our proposed solution for simplicity, and later we discuss how to extend the solution to LambdaRank.

We assume that there exists an unknown function $h(\cdot)$ that models the relevance quality of document $\xb$ under the given query $q$ as $h(\xb)$. In order to learn this function, we utilize a fully connected neural network $f(\xb;\btheta) = \sqrt{m}\Wb_L \phi(\Wb_{L-1} \phi(\dots \phi(\Wb_1\xb))$, where depth $L \geq 2$, $\phi(\xb) = \max\{\xb, 0\}$, and $\Wb_1 \in \RR^{m \times d}$, $\Wb_i \in \RR^{m \times m}$, $2\leq i \leq L-1$, $\Wb_L \in \RR^{m \times 1}$, and $\btheta = [\text{vec}(\Wb_1)^\top,\dots,\text{vec}(\Wb_L)^\top]^\top \in \RR^{p}$ with $p = m+md+m^2 (L-2)$. Without loss of generality, we assume the width of each hidden layer is the same as $m$, concerning the simplicity of theoretical analysis. We also denote the gradient of the neural network function as $\gb(\xb; \btheta) = \nabla_{\btheta} f(\xb; \btheta) \in \RR^p$. 

RankNet specifies a distribution on pairwise comparisons. In particular, the probability that document $i$ is more relevant than document $j$ is calculated by $\PP(i \succ j) = \sigma(f(\xb_i; \btheta) - f(\xb_j; \btheta))$, where 
$\sigma(s) = 1 / (1 + \exp(-s))$. For simplicity, we use $f_{ij}^t$ to denote $f(\xb_i;\btheta_{t-1}) - f(\xb_j; \btheta_{t-1})$. Therefore, the objective function for $\btheta$ estimation in RankNet can be derived under a cross-entropy loss between the predicted pairwise comparisons and those inferred from user feedback till round $t$ and a L2-regularization term centered at the randomly initialized parameter $\btheta_0$:
\small
\begin{align}
\label{eq:loss}
   \cL_t(\btheta) =& \sum\nolimits_{s=1}^t\sum\nolimits_{(i, j) \in \Omega_s} -(1 - \yijs)\log(1 - \sigma(f_{ij})) \nonumber \\
     &-  \yijs\log(\sigma(f_{ij})) + {m \lambda}/{2}\|\btheta - \btheta_0\|^2,
\end{align}
\normalsize
where 
$\lambda$ is the L2 regularization coefficient, $\Omega_s$ denotes the set of document pairs that received different click feedback at round $s$, i.e. $\Omega_s = \{(i, j): c_i^s \neq c_j^s, \forall \tau_s(i) \leq \tau_s(j) \leq o_t\}$, $\yijs$ indicates whether document $i$ is preferred over document $j$ in the click feedback, i.e., $\yijs = (c_i^s - c_j^s)/2 + 1/2$~\citep{burges2010ranknet}. 

The online estimation of RankNet boils down to the construction of $\{\Omega_t\}_{t=1}^T$ over time. However, the conventional practice of using all the inferred pairwise preferences from clicks becomes problematic in an online setting. For example, in the presence of click noise (e.g., a user mistakenly clicks on an irrelevant document), pairing documents would cause a quadratically increasing number of noisy training instances, and therefore impose a strong negative impact on the quality of the learned ranker and subsequent result serving. To alleviate this deficiency, we propose to only use \emph{independent} pairwise comparisons to construct the training set, e.g., $\Omega_t^{ind} = \{(i, j): c_i^t \neq c_j^t, \forall (\tau_t(i), \tau_t(j)) \in D\}$, where $D$ represents the set of disjointed position pairs, for example, $D = \{(1, 2), (3, 4), ... (o_t-1, o_t)\}$. In other words, we only use a subset of non-overlapping pairwise comparisons for update.

\begin{figure}
    \centering
    \includegraphics[width=0.88\linewidth]{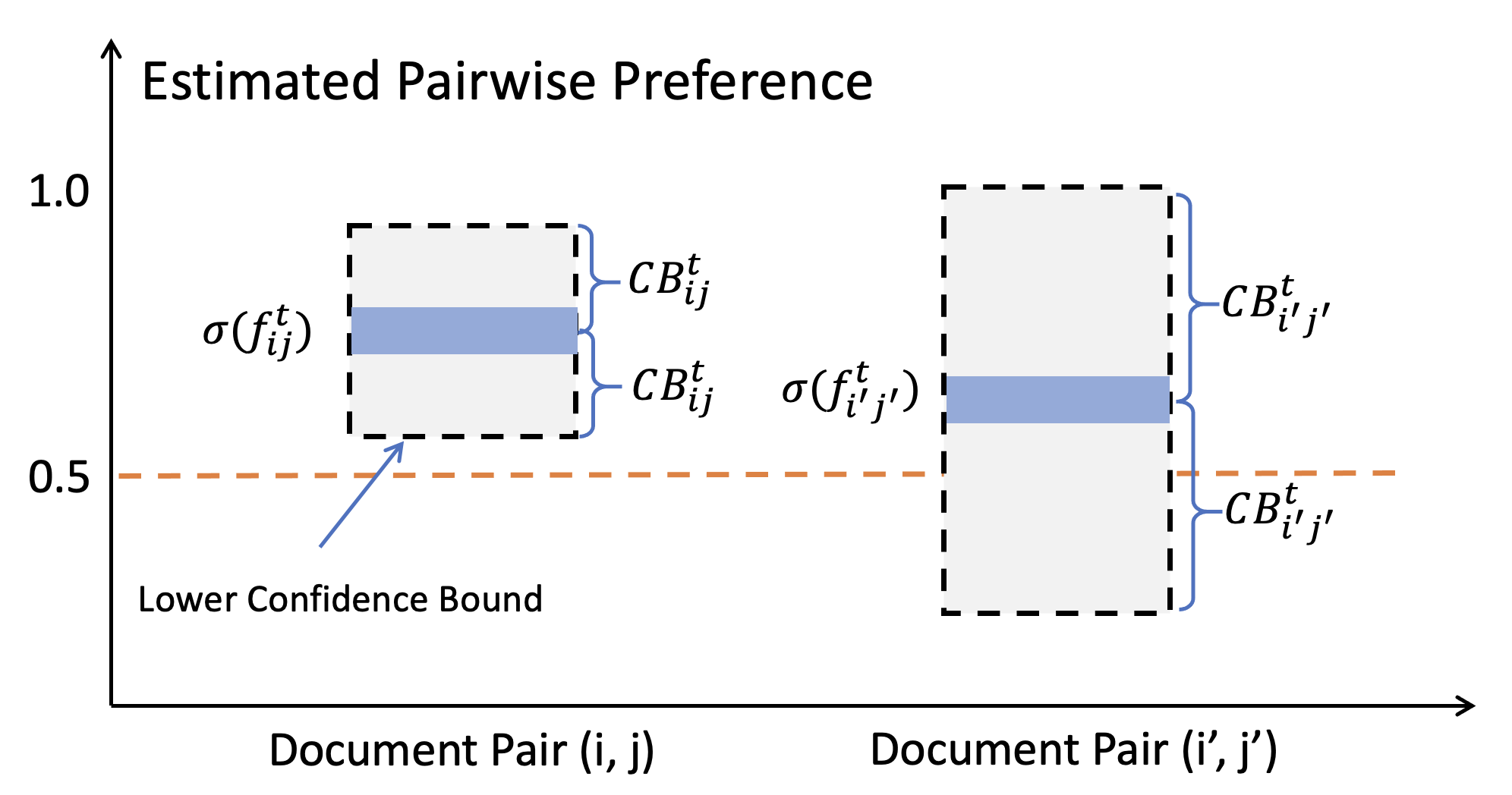}
    \vspace{-3mm}
    \caption{Illustration of certain and uncertain rank orders.}
    \label{fig:certain_rank}
    \vspace{-5mm}
\end{figure}


\noindent\textbf{Result Ranking Strategy.}
Another serious issue in the online collected training instances is bias. As discussed before, the ranking model is updated based on the acquired feedback from what it has presented to the users so far, which is subject to various types of biases, e.g., presentation bias and position bias \citep{joachims2005accurately,joachims2007evaluating,agichtein2006improving}. Hence, it is vital to \textit{effectively explore the unknowns} to complete the ranker's knowledge about the ranking space, while \textit{serving users with qualified ranking results} to minimize regret. As our solution of result ranking, we explore in the pairwise document ranking space with respect to the ranker's current uncertainty about the comparisons. 

To quantify the source of uncertainty, we follow conventional click models to assume that on the \emph{examined} documents where $\tau_t(i) \leq o_t$, the obtained feedback $C_t$ is independent from each other given the \emph{true relevance} of documents, so is their noise \citep{joachims2005accurately,guo2009click,guo2009efficient}. As a result, the noise in each collected preference pair becomes the sum of noise from the clicks in the two associated documents. Because we only use the independent pairs $\Omega_t^{ind}$, the pairwise noise is thus independent of each other and the history of result serving, which leads to the following proposition.

\begin{proposition}
\label{prop:pairwise}
For any $t \geq 1$, $\forall (i, j) \in \Omega_t^{ind}$, the pairwise feedback follows $y_{ij}^t = \sigma(h(\xb_i) - h(\xb_j)) + \xi_{ij}^t$, where $\xi_{ij}^t$ satisfying that for all $\beta \in \RR$, $\EE[\exp(\beta\xi_{ij}^t) | \{\{\xi^s_{\ijp}\}_{\ijp \in \Omega_s^{ind}}\}_{s=1}^{t-1}, \Omega_{1:t-1}^{ind}] \leq \exp(\beta^2\nu^2)$, is a $\nu$-sub-Gaussian random variable
\end{proposition}


Based on the property of sub-Gaussian random variables, the proposition above can be easily satisfied in practice as long as the pointwise click noise follows a sub-Gaussian distribution. Typicall the pointwise noise is modeled as a binary random variable related to the document's true relevance under the given query, which  follows a $\frac{1}{2}$-sub-Gaussian distribution. 
Let $\Psi_t$ represent the set of all possible document pairs at round $t$, e.g., $\Psi_t = \{(i, j) \in [V_t]^2, i \neq j\}$ and $|\Psi_t| = V_t^2 - V_t$. Based on the objective function Eq \eqref{eq:loss} over training dataset $\{\Omega_s^{ind}\}_{s=1}^t$, we have the following lemma bounding the uncertainty of the estimated pairwise rank order at round $t$.

\begin{lemma}(Confidence Interval of Pairwise Rank Order). 
\label{lemma_CI}
There exist positive constants $C_1$ and $C_2$ such that for any $\delta_1 \in (0,1)$, if the step size of gradient descent $\eta \leq C_1(TmL + m\lambda)^{-1}$ and 
$m \geq C_2\max\big\{ \lambda^{-1/2}L^{-3/2}(\log(TV_{\max}L^2/\delta_1))^{3/2}, T^7\lambda^{-7}L^{21}(\log m)^3\big\}$,
then at round $t < T$, for any document pair $(i, j) \in \Psi_t$ under query $q_t$, with probability at least $1 - \delta_1$,
\small
\begin{equation}
    |\sigma(f_{ij}^t) - \sigma(h_{ij}) | \leq \alpha_t\Vert\gb^t_{ij}/\sqrt{m}\Vert_{\mathbf{\Ab}_t^{-1}} + \epsilon(m),
\end{equation}
\normalsize
where $V_{\max}$ represents the maximum number of documents under a query over time, \small$\epsilon(m) = \bar C_3\Big(T^{2/3}m^{-1/6}\lambda^{-2/3}L^3\sqrt{\log (m)} + L^{1/2}(1 - \eta m \lambda)^{J/2}\sqrt{T/\lambda} + T^{5/3}m^{-1/6}\lambda^{-5/3}L^4(1 + \sqrt{T/\lambda})\Big)$,  $h_{ij} = h(\xb_i) - h(\xb_j)$, $\gb_{ij}^s = \gb(\xb_i; \btheta_s) - \gb(\xb_j; \btheta_s)$, $\Ab_t = \sum_{s=1}^{t-1}\sum_{(i^\prime, j^\prime) \in \Omega^{ind}_{s}}
\frac{1}{m}\gb_{\ijp}^s{\gb_{\ijp}^s}^\top + \lambda \mathbf{I}$, $\alpha_t = \bar C_1\Big(\sqrt{\nu^2\log ({\det(\Ab_t)}/{\delta_1^2 \det(\lambda \Ib))}} + \sqrt{\lambda}{\bar C_2}\Big)$, \normalsize$\bar C_1$, $\bar C_2$ and $\bar C_3$ are positive constants.
\label{lemma:cb}
\end{lemma}


We provide the detailed proof of Lemma~\ref{lemma:cb} and the specification of constants $\{C_1, C_2, \bar C_1, \bar C_2, \bar C_3\}$ in the appendix. This lemma provides a tight high probability bound of the pairwise rank order estimation uncertainty under RankNet. The uncertainty caused by the variance from the pairwise observation noise is controlled by $\alpha_t$, and $\epsilon(m)$ is the approximation error incurred in the estimation of the true scoring function. This enables us to perform efficient exploration in the pairwise document ranking space for the model update. To illustrate our ranking strategy, we introduce the following notion on the estimated pairwise preference.

\begin{definition} (Certain Rank Order)
\label{def:certain}
At round $t$, the rank order between documents $(i, j) \in \Psi_t$ is in a certain rank order if and only if $\sigma(f_{ij}^t) - CB_{ij}^t > \frac{1}{2}$, where $CB_{ij}^t=\alpha_{t}\Vert\gb_{ij}^t/\sqrt{m}\Vert_{\Ab_t^{-1}} - \epsilon(m)$ is the width of confidence bound about the estimated pairwise rank order.
\end{definition}

\begin{algorithm}[t]
	\caption{Online Neural Ranking Algorithm} 
	\label{algorithm:NR} 
	\begin{algorithmic}[1]
    \STATE  \textbf{Input:} L2 coefficient $\lambda$, step size $\eta$, number of iterations for gradient descent $J$, network width $m$, network depth $L$.
	\STATE Initialize $\btheta_0 = (\text{vec}(\Wb_1), \dots \text{vec}(\Wb_L)) \in \RR^p$, where for each $1 \leq l \leq L-1$, $\Wb_l = (\Wb, \zero; \zero, \Wb)$, each entry of $\Wb$ is initialized independently from $N(0, 4/m)$; $\Wb_L = (\wb^\top, -\wb^\top)$, where each entry of $\wb$ is initialized independently from $N(0, 2/m)$.
	\STATE Initialize $A_1 = \lambda \Ib$
	\FOR{$t=1, \dots, T$}
	\STATE $q_t \leftarrow receive\_query(t)$ 
	\STATE $\cX_t = \{\xb_1^t, \cdots, \xb_{n_t}^t\} \leftarrow retrieve\_documents(q_t)$
	\STATE $\omega_t \leftarrow construct\_certain\_rank\_order\_set(\cX_t, \btheta_{t-1}, A_t)$
    \STATE $\tau_t \leftarrow topological\_sort(\omega_t)$ 
	\STATE $C_t \leftarrow collect\_click\_feedback(\tau_t)$
	\STATE $\Omega_t^{ind} \leftarrow construct\_independent\_pairs(C_t)$
    \STATE Set $\btheta_t$ to be the output of gradient descent with step size $\eta$ for $J$ rounds on:\\
    $\btheta_t = \argmin_{\btheta} \sum_{s=1}^t\sum_{(i, j) \in \Omega_s^{ind}} -(1 - \yijs)\log(1 - \sigma(f_{ij})) -  \yijs\log(\sigma(f_{ij})) + ({m \lambda}/{2})\|\btheta - \btheta_0\|^2$
    \STATE $\Ab_{t+1} = \Ab_t + \sum_{(i, j)\in \Omega_t^{ind}} \gb_{ij}^t{\gb_{ij}^t}{^\top}/m$
	\ENDFOR
	\end{algorithmic}
\end{algorithm}

Based on Lemma~\ref{lemma:cb}, if an estimated rank order $(i \succ j)$ is a certain rank order, with a high probability that the estimated preference is consistent with the ground-truth. Hence, they should be followed in the returned ranked list. For example, as shown in Figure~\ref{fig:certain_rank}, 
the lower bound for $\sigma(f_{ij}^t)$ estimation is larger than $\frac{1}{2}$, which indicates consistency between the estimated and ground-truth rank order between $(i,j)$. 
But with $\sigma(f_{\ijp}^t)-CB_{\ijp}^t < 1/2$, the estimated order $(i^\prime {\succ} j^\prime)$ is still uncertain as the ground-truth may present an opposite order. 


We use $\omega_t$ to represent the set of all certain rank orders at round $t$, $\omega_t = \{(i, j) \in \Psi_t: \sigma(f_{i, j}^t) - CB_{i, j}^t > \frac{1}{2}\}$. 
For pairs in $\omega_t$, we can directly exploit the current estimated rank order as it is already consistent with the ground-truth. But, for the uncertain pairs that do not belong to $\omega_t$, exploration is necessary to obtain feedback for further model update (and thus to reduce uncertainty). For example, in the document graph shown in Figure~\ref{fig:model}, when generating the ranked list, we should exploit the current model by preserving the order between document A and documents B, C, D, while randomly swap the order between documents (B, C), (C, D), (E, F) to explore (in order to conquer feedback bias). 

The estimated pairwise rank order, $\sigma(f_{ij}^t)$, is derived based on relevance score calculated by the current neural network, i.e., $f(\xb_i; \btheta_{t-1})$ and $f(\xb_j; \btheta_{t-1})$. Hence, as shown in Figure~\ref{fig:model}, due to the monotonicity and transitivity of the sigmoid function, the document graph constructed with the candidate documents as the vertices and the certain rank order as the directed edges is a directed acyclic graph (DAG). We can perform a topological sort on the constructed document graph to efficiently generate the final ranked list. The certain rank orders are preserved by topological sort to exploit the ranker's high confidence predictions. On the other hand, the topological sort randomly chooses vertices with zero in-degree, among which there is no certain rank orders. This naturally achieves exploration among uncertain rank orders.
In Figure~\ref{fig:model}, as document A is predicted to be better than all the other documents by certain rank orders, it will be first added to the ranked list and removed from the document graph by topological sort. In the updated document graph, both document B and C become vertices with zero in-degree as the estimated rank order between them is still uncertain. Topological sort will randomly choose one of them as the next document in the ranked list, which induces exploration on the uncertain rank orders. Two possible ranked lists are shown in the figure. As exploration is confined to the pairwise ranking space, it effectively reduces the exponentially sized exploration space of result ranking to quadratic. Algorithm~\ref{algorithm:NR} shows the details of the proposed solution.


\noindent\textbf{Extend to LambdaRank.}
LambdaRank directly optimizes the ranking metric of interest (e.g., NDCG) with a modified gradient based on RankNet \citep{burges2010ranknet}. For a given pair of documents, the confidence interval of LambdaRank's estimation can be calculated by gradients of the neural network in the same way as in RankNet (i.e., by Lemma \ref{lemma_CI}). However, as the objective function of LambdaRank is unknown, it prevents us from theoretically analyzing the resulting online algorithm's regret. But similar empirical improvement from LambdaRank against RankNet known in the offline settings \citep{burges2010ranknet} is also observed in our online versions of these two algorithms. 
\section{Regret Analysis}
\label{sec:regret}
Our regret analysis is built on the latest theoretical studies in deep neural networks. Recent attempts show that in the neural tangent kernel (NTK) space, the generalization error bound of a DNN can be characterized by the corresponding best function~\citep{cao2019generalization1, cao2019generalization2, chen2019much, daniely2017sgd, arora2019exact}. In our analysis, we denote the NTK matrix of all possible pairwise document tangent features as $\Hb \succeq \lambda_0\Ib$, with the effective dimension of $\Hb$ denoted as $\tilde{d}$. 
Due to limited space, we leave the detailed definition of $\Hb$ and $\tilde{d}$ in the appendix.

We define event $E_t$ as: $E_t = \big\{ \forall (i, j) \in \Psi_t, |\sigma({f_{ij}^t}) - \sigma(h_{ij}) | \leq CB_{i, j}^t\big\}$ at round $t$. $E_t$ suggests that the estimated pairwise rank order on all the candidate document pairs under query $q_t$ is close to the ground-truth at round $t$. According to Lemma~\ref{lemma:cb}, it is easy to reach the following conclusion,

\begin{corollary}On the event $E_t$, it holds that $\sigma(h_{ij}) > \frac{1}{2}$ if $(i, j) \in \omega_t$, i.e., in a certain rank order.
\label{col}
\end{corollary} 
Based on the definition of pairwise regret in Eq \eqref{eq_regret_def}, the ranker only suffers regret as a result of misplacing a pair of documents, i.e., swapping a pair into an incorrect order. According Corollary~\ref{col}, under event $E_t$, the certain rank order identified is consistent with the ground-truth.
As in our proposed solution, the certain rank order is preserved by the topological sort, it is easy to verify that regret only occurs on the document pairs with uncertain rank order. 
Therefore, the key step in our regret analysis is to count the expected number of uncertain rank orders. According to Definition~\ref{def:certain}, a pairwise estimation is certain if and only if $|\sigma(f_{ij}^t) - \frac{1}{2}| \geq CB_{i, j}^t$. Hence, we have the following lemma bounding the probability that an estimated rank order being uncertain.

\begin{lemma}
With $\eta$, $m$ satisfying the same conditions in Lemma~\ref{lemma:cb}, with $\delta_1 \in (0, 1/2)$ defined in Lemma~\ref{lemma:cb}, and $\delta_2 \in (0, 1/2)$, such that for $t \geq t^\prime = \mathcal{O}(\log(1/\delta_2) + \log(1/\delta_1))$, under event $E_t$, the following holds with probability at least $ 1 - \delta_2$:
\begin{align*}
    \forall (i, j) \in \Psi_t, \PP((i, j) \notin \omega_t) \leq \frac{C_u\log (1 / \delta_1)}{(\Delta_{\min} - 2\epsilon(m))^2}\|\gb_{ij}^t/\sqrt{m}\|_{\Ab_t^{-1}}^2 ,
\end{align*}
where $C_u = 8\nu^2k_{\mu}^2/c_{\mu}^2$ with $k_{\mu}$ and $c_{\mu}$ as the Lipschitz constants for the sigmoid function, $\Delta_{\min} = \min\limits_{t\in T, (i, j) \in \Psi_t}| \sigma(h_{ij}) - \frac{1}{2}|$ represents the smallest gap of pairwise difference between any pair of documents under the same query over time.
\label{lemma:uncertain}
\end{lemma}

\begin{remark}
With m satisfying the condition in Lemma \ref{lemma:cb}, and setting the corresponding $\eta$ and $J = \tilde{\mathcal{O}}(TL/\lambda)$, $\epsilon(m)  = O(1)$ can be achieved. More specifically, there exists a positive constant $c$ such that $\Delta_{\min}-2\epsilon(m) = c\Delta_{\min}$.
\end{remark}

Lemma~\ref{lemma:uncertain} gives us a tight bound for an estimated pairwise order being uncertain. Intuitively, it targets to obtain a tighter bound on the uncertainty of the neural model's parameter estimation compared to the bound determined by $\delta_1$ in Lemma~\ref{lemma:cb}. With this bound, the corresponding confidence interval will exclude the possibility of flipping the estimated rank order, i.e., the lower confidence bound of this pairwise estimation is above 0.5. 

In each round of result serving, as the model $\btheta_t$ will not change before the next round starts, the expected number of uncertain rank orders, denoted as $\EE[U_t]$, can be estimated by the summation of the uncertain probabilities over all possible pairwise comparisons under the query $q_t$, e.g., $\EE[U_t] = \frac{1}{2} \sum_{(i, j) \in \Psi_t} \mathbb{P}((i, j) \notin \omega_t)$.
Denote $p_{t}$ as the probability that the user examines all documents in $\tau_t$ at round $t$, and let $p^* = \min_{1\leq t \leq T} p_{t}$ be the minimal probability that all documents in a query are examined over time. We present the upper regret bound as follows.
\begin{theorem}
\label{thm:upper-regret}
With $\delta_1$ and $\delta_2$ defined in Lemma~\ref{lemma:cb}, \ref{lemma:uncertain}, $\eta$, $m$ satisfying the same conditions in Lemma~\ref{lemma:cb},
there exist positive constants $\{C_i^r\}_{i=1}^2$ that
with probability at least $1 - \delta_1$, the $T$-step regret is bounded by:
\small
\begin{align*}
    R_T 
     \leq& R^{\prime} + (C^r_1\log(1/\delta_1)\tilde{d}\log(1 + TV_{\max}/\lambda) + C^r_2)(1 - \delta_2)/(\Delta_{\min}^2p^*)
\end{align*}
\normalsize
where $R^{\prime} = t^{\prime}V_{\max}^2 + (T - t^{\prime})\delta_2V_{\max}^2$, with $t^{\prime}$ and $V_{\max}$ defined in Lemma~\ref{lemma:uncertain}.
By choosing $\delta_1 = \delta_2 = 1/T$, the expected regret is at most $O(\tilde{d}\log^2(T))$.
\end{theorem}

\begin{hproof}
The detailed proof is provided in the appendix. We only provide the key ideas behind our regret analysis here.
The regret is first decomposed into two parts. First, $R^\prime$ represents the regret when Lemma~\ref{lemma:uncertain} does not hold, in which the regret is out of our control. We use the maximum number of pairs associated with a query over time, i.e., $V_{\text{max}}^2$, to upper bound it. The second part corresponds to the cases when Lemma~\ref{lemma:uncertain} holds. Then, the instantaneous regret at round $t$ can be bounded by
$r_t = \mathbb{E} \big[K(\tau_t, \tau_t^*)\big] \leq \EE[U_t]$, as only the uncertain rank orders would induce regret.
\end{hproof}

In this analysis, we provide a gap-dependent regret upper bound, where the gap $\Delta_{\min}$ characterizes the intrinsic difficulty of sorting the $V_t$ candidate documents at round $t$. Intuitively, when $\Delta_{\min}$ is small, e.g., comparable to the network's resolution $\epsilon(m)$, many observations are needed to recognize the correct rank order between two documents. As the matrix $\Ab_t$ only contains information from examined document pairs,  our algorithm guarantees that the cumulative pairwise regret of the examined documents until round $t$ ( $\{1:o_s\}_{s=1}^t$) to be sub-linear, while the regret in the leftover documents ($\{o_s+1:V_s\}_{s=1}^t$) is undetermined. We adopt a commonly used technique that leverages the probability that a ranked list is fully examined to bound the regret on those unexamined documents~\citep{li2016contextual, kveton2015combinatorial, kveton2015tight}. This probability is a constant independent of $T$.
It is worth noting that our algorithm does not need the knowledge of $p^*$ for model learning or result ranking; it is solely used for the regret analysis to handle the partial observations. From a practical perspective, the ranking quality of documents ranked below $o_s$ for $s \in [T]$ does not affect users' online experience, as the users do not examine them. Hence, if we only count regret in the examined documents, $R_T$ does not need to be scaled by $p^*$

\begin{remark}
Our regret is defined over the number of mis-ordered pairs, which is the \emph{first} pairwise regret analysis for a neural OL2R algorithm. Existing OL2R algorithms optimize their own metrics (e.g., utility function as defined in \citep{yue2009interactively}), which can hardly link to any conventional ranking metrics. As shown in \citep{Wang2018Lambdaloss}, most classical ranking evaluation metrics, such as NDCG, are based on pairwise document comparisons. Our regret analysis connects our OL2R solution's theoretical property with such metrics, which is also confirmed in our empirical evaluations.   
\end{remark}
\section{Experiments}

\label{sec:exp}
In this section, we empirically compare our proposed models with an extensive list of state-of-the-art OL2R algorithms on two large public learning to rank benchmark datasets. We implemented all the neural rankers in PyTorch and performed all the experiments on a server equipped with Intel Xeon Gold 6230 2.10GHz CPU, 128G RAM, four NVIDIA GeForce RTX 2080Ti graphical cards.

\subsection{Experiment Setup}

\noindent\textbf{Datasets.} We experiment on two publicly available learning to rank datasets, Yahoo! Learning to Rank Challenge dataset \citep{chapelle2011yahoo}, which consists of 292,921 queries and 709,877 documents represented by 700 ranking features, and MSLR-WEB10K \citep{qin2013introducing}, which contains 30,000 queries, each having 125 documents on average  represented by 136 ranking features. Both datasets are labeled on a five-grade relevance scale: from not relevant (0) to perfectly relevant (4). We followed the train/test/validation split provided in the datasets to make our results comparable to the previously reported results.

\noindent\textbf{Non-linearity analysis.}
Most of the existing OL2R models assume that the expected relevance of a document under the given query can be characterized by a linear function in the feature space. However, such an assumption often fails in practice, as the potentially complex non-linear relations between queries and documents are ignored. 
For example, classical query-document features are usually constructed in parallel to the design and choices of ranking models. As a result, a lot of correlated and sometimes redundant features are introduced for historical reasons; and the ranker is expected to handle it. For instance, the classical keyword matching based features, such as TF-IDF, BM25 and language models, are known to be highly correlated~\citep{fang2004formal}; and the number of in-links is also highly related to the PageRank feature. 

To verify this issue, we performed a linear discriminative analysis (LDA)~\citep{balakrishnama1998linear} on both datasets. The technique of LDA is typically used for multi-class classification that automatically performs dimensionality reduction, providing a projection of the dataset that can best linearly separate the samples by their assigned class. We provide the entire labeled dataset for the algorithm to learn the separable representation. We set the reduced dimension to be two to visualize the results. In Figure~\ref{fig:lda}, we can clearly observe that a linear model is insufficient to separate the classes in both datasets.

\begin{figure*}[t]
  \centering
  \begin{subfigure}[b]{0.45\textwidth}
    \centering
    \includegraphics[width=\linewidth]{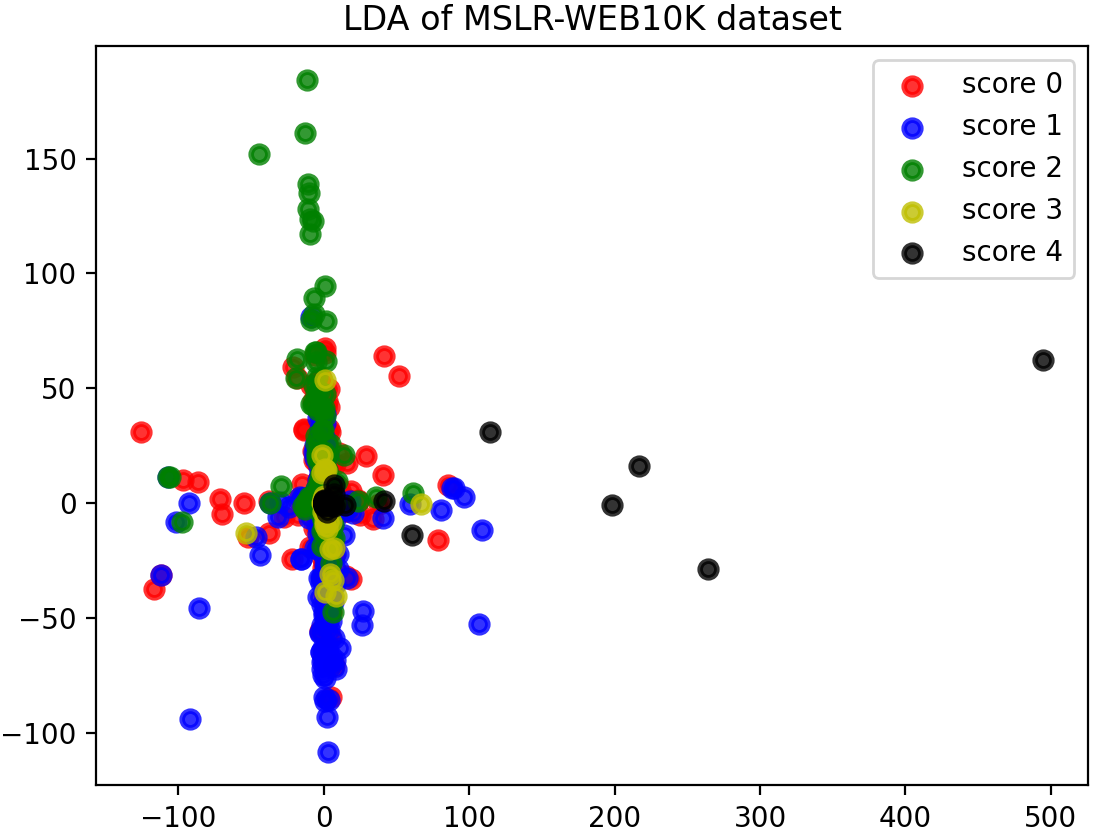}
  \end{subfigure}
  \begin{subfigure}[b]{0.45\textwidth}
    \centering
    \includegraphics[width=0.97\linewidth]{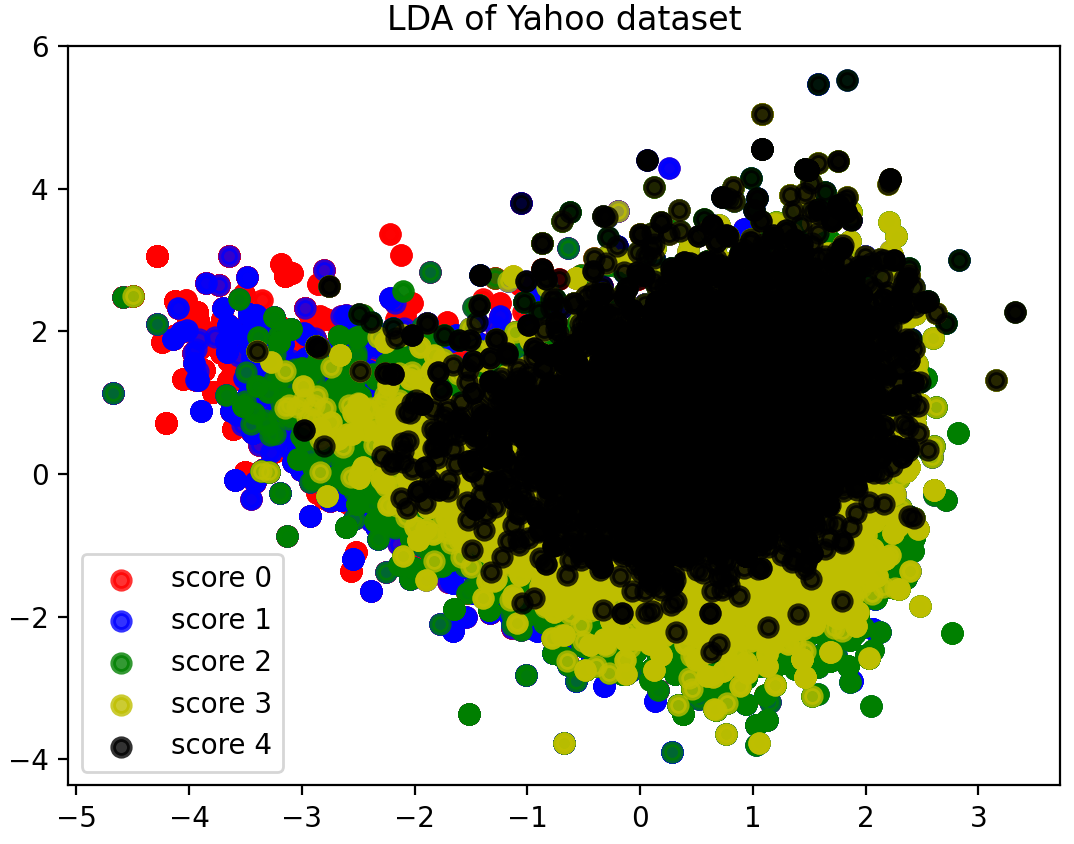}
  \end{subfigure}
  \caption{LDA based non-linearity analysis on both datasets.}\label{fig:lda} 
\end{figure*}

\noindent\textbf{User interaction simulation.} For reproducibility, user clicks are simulated via the standard procedure for OL2R evaluations \citep{oosterhuis2018differentiable}. At each round, a query is uniformly sampled from the training set for result serving. Then, the model determines the ranked list and returns it to the user. User click is simulated with a dependent click model (DCM)~\citep{guo2009efficient}, which assumes that the user will sequentially scan the list and make click decisions on the examined documents. In DCM, the probabilities of clicking on a given document and stopping examination are both conditioned on the document's true relevance label. We employ three different model configurations to represent three different types of users, for which details are shown in Table \ref{table:click}. Basically, we have the \textit{perfect} users, who click on all relevant documents and do not stop browsing until the last returned document; the \textit{navigational} users, who are very likely to click on the first encountered highly relevant document and stop there; and the \textit{informational} users, who tend to examine more documents, but sometimes click on irrelevant documents, such that contributing a significant amount of noise in their click feedback. To reflect presentation bias, all models only return the top 10 ranked results.

\small

\begin{table}[!htbp]
\vspace{-2mm}
  \caption{Configuration of simulated click models.}
  \vspace{-2mm}
  \label{table:click}
  \centering
  
  \begin{tabular}{cccccc|ccccc}
    \hline
                & \multicolumn{5}{c}{Click Probability} & \multicolumn{5}{c}{Stop Probability} \\
R & 0           & 1          & 2   &3 &4        & 0          & 1          & 2   &3 &4       \\ \hline
\textit{per}         & 0.0         & 0.2        & 0.4 &0.8 &1.0        & 0.0        & 0.0        & 0.0   &0.0 &0.0     \\
\textit{nav}    & 0.05   &0.3      & 0.5    &0.7    & 0.95       & 0.2       &0.3 & 0.5    &0.7    & 0.9        \\
\textit{inf}   & 0.4      &0.6   & 0.7   &0.8     & 0.9        & 0.1      &0.2  & 0.3     &0.4   & 0.5        \\ \hline
\end{tabular}
\vspace{-2mm}
\end{table}
\normalsize

\noindent\textbf{Baselines.} We list the OL2R solutions used for our empirical comparisons below. And we name our proposed model as olRankNet and olLambdaRank in the experiment result discussions.
\begin{itemize}
    \item \textbf{$\epsilon$-Greedy} \citep{hofmann2013balancing}: At each position, it randomly samples an unranked document with probability $\epsilon$ or selects the next best document based on the currently learned RankNet.
    \item \textbf{Linear-DBGD and Neural-DBGD} \citep{yue2009interactively}: DBGD uniformly samples a direction from the entire model space for exploration and model update. We apply it to both linear and neural rankers.
    \item \textbf{Linear-PDGD and Neural-PDGD} \citep{oosterhuis2018differentiable}: 
    PDGD samples the next ranked document from a Plackett-Luce model and estimates gradients from the inferred pairwise preferences. We also apply it to both linear and neural network rankers.
    \item \textbf{PairRank} \citep{jia2021pairrank}: This is a recently proposed OL2R solution based on a pairwise logistic regression ranker. As it is designed for logistic regression, it cannot be used for learning a neural ranker.
    \item \textbf{olLambdaRank GT}: At each round, we estimate a new LambdaRank model with ground-truth relevance labels of all the presented queries. This serves as the skyline in all our experiments.
\end{itemize}

\begin{figure*}[t]
  \centering
  \begin{subfigure}[b]{\textwidth}
    \centering
    \includegraphics[width=\linewidth]{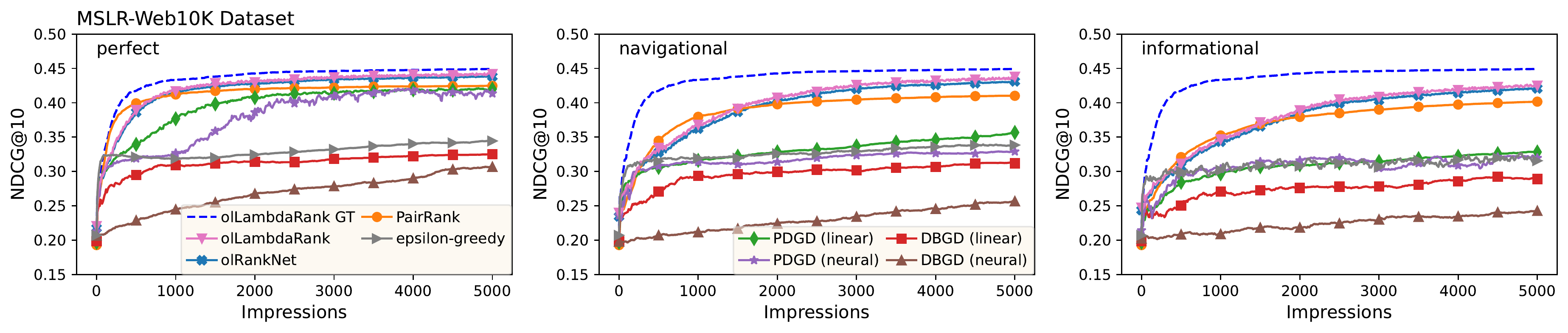}
  \end{subfigure}
  \begin{subfigure}[b]{\textwidth}
    \centering
    \includegraphics[width=\linewidth]{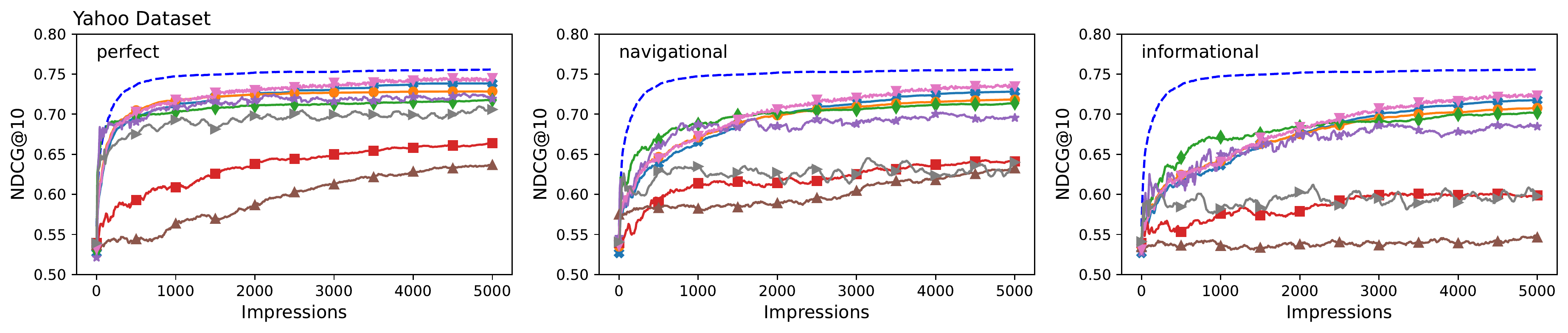}
  \end{subfigure}
  \caption{Offline performance on two benchmark datasets under three different click model configurations.}\label{fig:offline_result} 
\end{figure*}

\noindent\textbf{Hyper-Parameter Tuning.} MSLR-WEB10K and Yahoo Learning to Rank dataset are equally partitioned into five folds, of which three parts are used for for training, one part for validation and and one part test. We did cross validation on each dataset. For each fold, the models are trained on the training set, and the hyper-parameters are selected based on the performance on the validation set. 

In the experiment, a two-layer neural network with width $m = 100$ is applied for all the neural rankers. We did a grid search for olRankNet and olLambdaRank for regularization parameter $\lambda$ over $\{10^{-i}\}_{i=1}^4$, exploration parameter $\alpha$ over $\{10^{-i}\}_{i=1}^4$, learning rate over $\{10^{-i}\}_{i=1}^3$. The same set of parameter tuning is applied for PairRank, except the model is directly optimized with L-BFGS. 
The model update in PDGD and DBGD is based on the optimal settings in their original paper. 
The hyper-parameters for PDGD and DBGD are the update learning rate and the learning rate decay, for which we performed a grid search for learning rate over $\{10^{-i}\}_{i=1}^3$, and the learning rate decay is set as 0.999977.

\subsection{Experiment Results}

\textbf{Offline performance.} The offline performance is evaluated in an ``online'' fashion: the newly updated ranker is immediately evaluated on a hold-out testing set against its ground-truth relevance labels. This measures how rapidly an OL2R model improves its ranking quality, and it is an important metric about users' instantaneous satisfaction. This can be viewed as using one portion of traffic for online model update, while serving another portion with the latest model. 
We use NDCG@10 to assess the ranking quality, and we compare all algorithms over three click models and two datasets. For online RankNet and online LambdaRank, since it is computationally expensive to store and operate on a complete $\Ab_t$ matrix, we only used its diagonal elements as an approximation. We fixed the total number of iterations $T$ to 5000. The experiments are executed for 10 times with different random seeds and the averaged results are reported in Figure~\ref{fig:offline_result}.

We can clearly observe that our proposed online neural ranking models achieved significant improvement compared to all baselines. Under different click models, both linear and neural DBGD performed the worst. This is consistent with previous findings: DBGD depends on interleave tests to determine the update direction in the model space. But such model-level feedback cannot inform the optimization of any rank-based metric. Moreover, with a neural ranker, random exploration becomes very ineffective. 
PDGD consistently outperformed DBGD under different click models. However, its document sampling based exploration limits its learning efficiency, especially when users only examine a small portion of documents, e.g., the navigational users. It is worth noting that in the original paper~\citep{oosterhuis2018differentiable}, PDGD with a neural ranker outperformed linear ranker after much more interactions, e.g., 20000 iterations. Our proposed solutions with only 5000 iterations already achieved better performance than the best results reported for PDGD, which demonstrates the encouraging efficiency of our proposed OL2R solution. 
Compared to PairRank, our neural rankers had a worse start at the beginning. We attribute it to the limited training samples available at the initial rounds, i.e., the network parameters were not well estimated yet. But the neural model enables non-linear relation learning and quickly leads to better performance than the linear models when more observations arrive. Compared to olRankNet, olLambdaRank directly optimizes the evaluation metrics, e.g., NDCG@10, with corresponding gradients. We can observe similar improvements from LambdaRank compared to RankNet as previously reported in offline settings.
It is worth noting that though the improvement of olRankNet and olLambdaRank compared to PairRank is not as large as their improvement against other baselines in the figure, small improvement in the performance metric often means a big leap forward in practice as most real-world systems serve millions of users, where even a small percentage improvement can be translated into huge utility gain to the population.

\begin{figure*}[t]
  \centering
  \begin{subfigure}[b]{\textwidth}
    \centering
    \includegraphics[width=\linewidth]{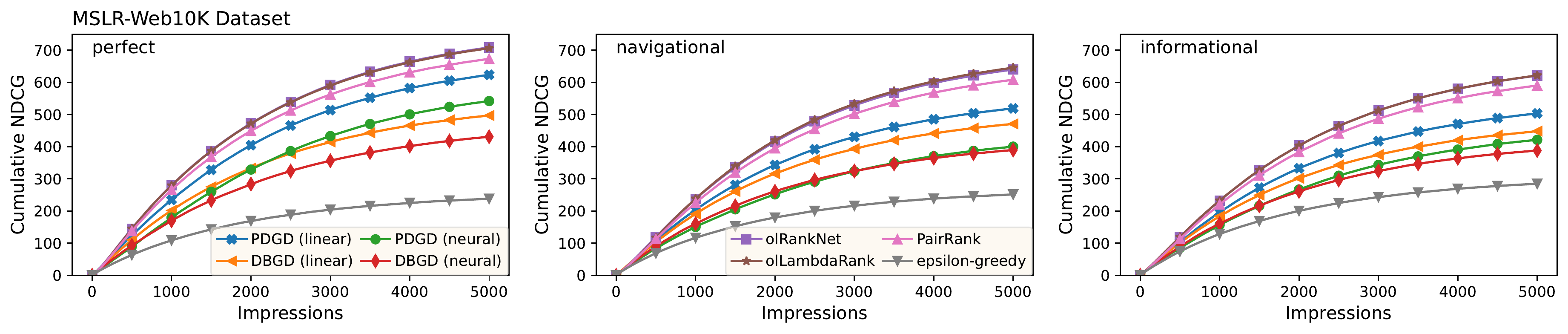}
  \end{subfigure}
  \begin{subfigure}[b]{\textwidth}
    \centering
    \includegraphics[width=\linewidth]{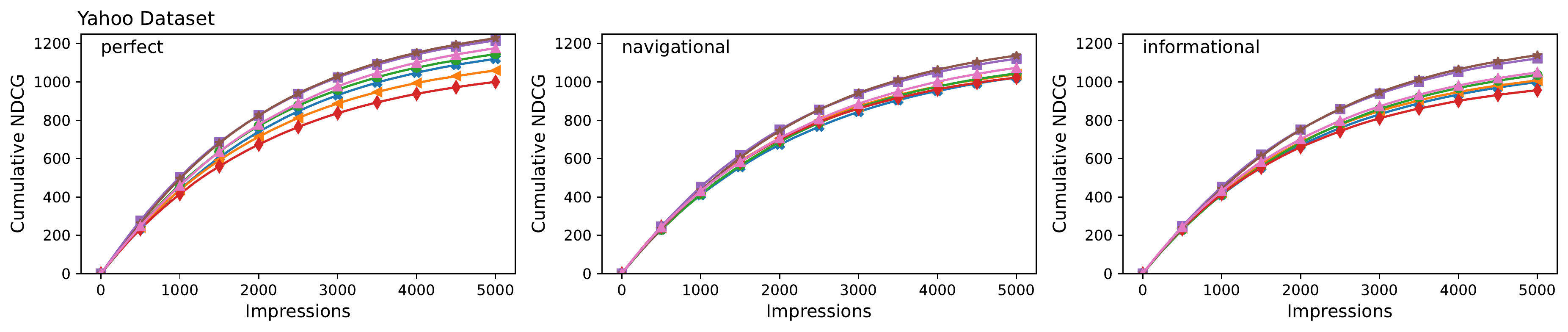}
  \end{subfigure}
  \caption{Online performance on two datasets under three different click model configurations.}
  \label{fig:online_result} 
\end{figure*}


\noindent\textbf{Online performance.} In OL2R, in addition to the offline evaluation, the models' ranking performance during online result serving should also be considered, as it reflects user experience during model update. Sacrificing users experience for model training will compromise the goal of OL2R. We adopt the cumulative Normalized Discounted Cumulative Gain to assess models' online performance. For $T$ rounds, the cumulative NDCG is calculated as
\small
\begin{equation*}
    \text{Cumulative NDCG} = \sum\nolimits_{t=1}^T \text{NDCG}(\tau_t) \cdot \gamma^{(t-1)},
\end{equation*}
\normalsize
which computes the expected utility a user receives with a probability $\gamma$ that he/she stops searching after each query~\citep{oosterhuis2018differentiable}. Following the previous work~\citep{oosterhuis2018differentiable, wang2019variance, wang2018efficient}, we set $\gamma = 0.9995$.

Figure~\ref{fig:online_result} shows the online performance of the proposed online neural ranking model and all the other baselines. It is clear to observe that DBGD-based models have a much slower convergence and thus have worse online performance. Compared to the proposed solution, PDGD showed consistently worse performance, especially under the navigational and informational click models with a neural ranker. We attribute this difference to the exploration strategy used in PDGD: PDGD's sampling-based exploration can introduce unwanted distortion in the ranked results, especially at the early stage of online learning. We should note the earlier stages in cumulative NDCG plays a much more important role due to the strong shrinking effect of $\gamma$.

Our proposed models demonstrated significant improvements over all baseline methods on both datasets under three different click models. Such improvement indicates the effectiveness our uncertainty based exploration, which only explores when the ranker's pairwise estimation is uncertain. Its advantage becomes more apparent in this online ranking performance comparison, as an overly aggressive exploration in the early stage costs more in cumulative NDCG. We can also observe the improvement of olLambdaRank compared to olRankNet in this online evaluation, although the difference is not very significant. The key reason is also the strong discount applied to the later stage of model learning: olLambdaRank's advantage in directly optimizing the rank metric becomes more apparent in the later stage, as suggested by the offline performance in Figure \ref{fig:offline_result}. At the beginning of model learning, both models are doing more explorations and therefore the online performance got more influenced by the number of document pairs with uncertain rank orders, rather than those with certain rank orders.


\begin{figure*}[t]
  \centering
  \begin{subfigure}[b]{0.47\textwidth}
    \centering
    \includegraphics[width=0.9\linewidth]{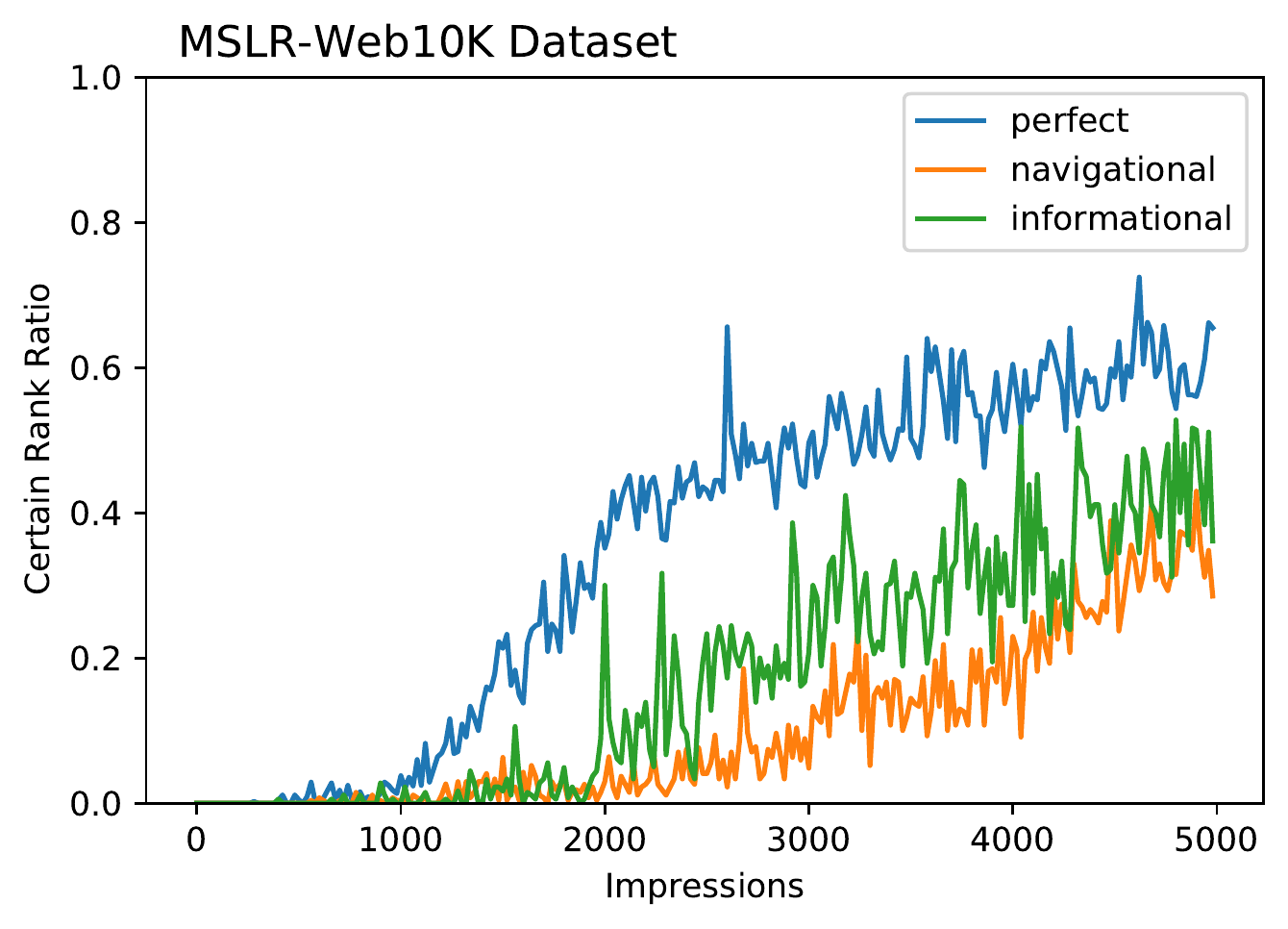}
    \label{fig:ratio_web10k}
  \end{subfigure}
  \begin{subfigure}[b]{0.47\textwidth}
    \centering
    \includegraphics[width=0.9\linewidth]{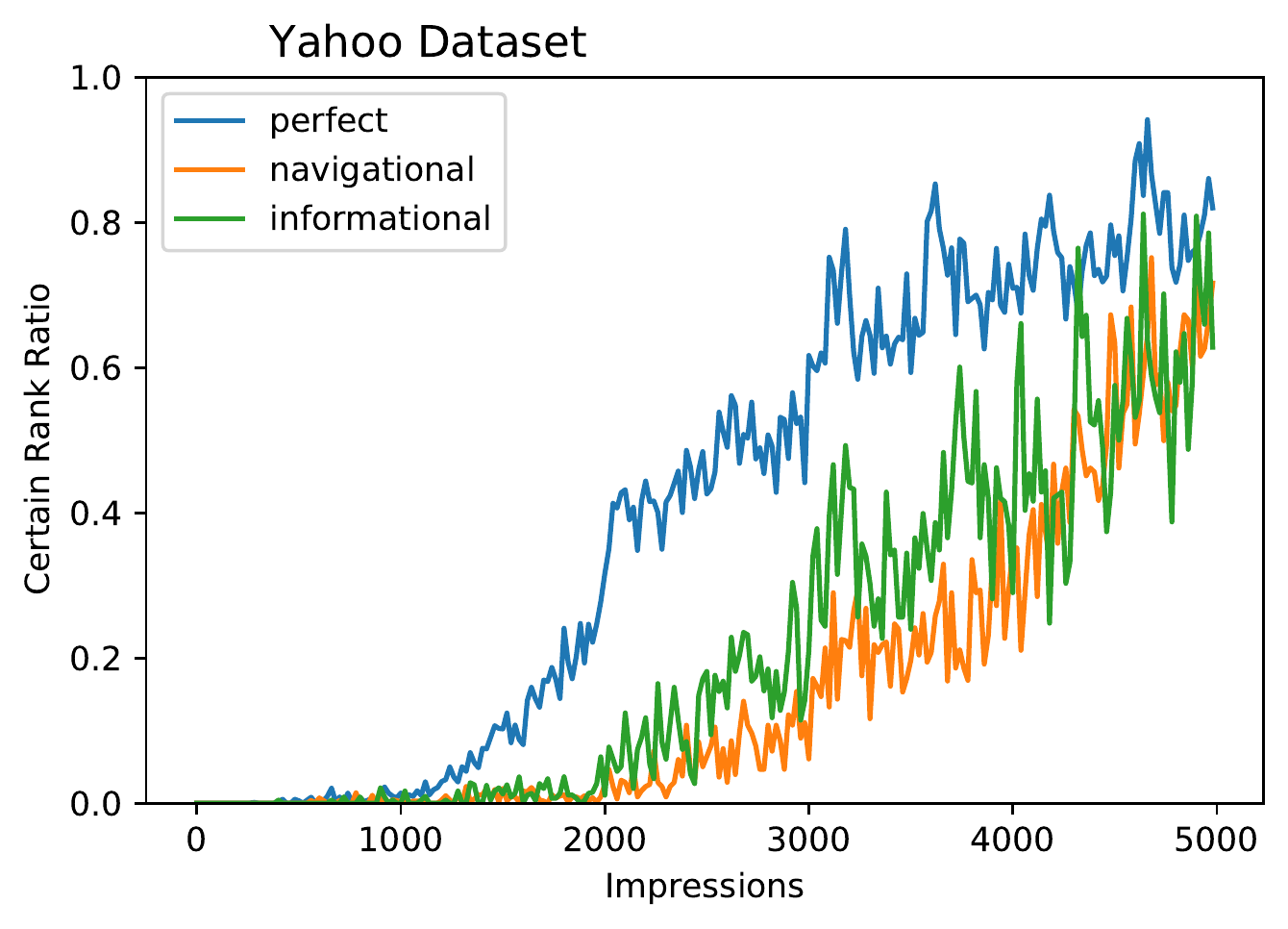}
    \label{fig:ratio_yahoo}
  \end{subfigure}
  \caption{Ratio of certain rank orders at Top-10 positions over the rounds of online model update.}\label{fig:ratio} 
\end{figure*}


\noindent\textbf{Shrinkage of the number of uncertain rank orders.} To further verify the effectiveness of the exploration strategy in our proposed online neural ranking model, we zoom into the trace of the number of identified certain rank orders under each query during online model update. As the model randomly shuffles the uncertain rank orders to perform the exploration, a smaller ratio of uncertain rank orders is preferred to reduce the regret, especially at the top ranked positions. Figure~\ref{fig:ratio} reports the ratio of certain rank orders among all possible document pairs at top-10 positions in our olRankNet model. We can clearly observe that the certain rank orders quickly reach a promising level, especially on the Yahoo dataset. This confirms our theoretical analysis about the convergence of the number of uncertain rank orders. Comparing the results under different click models, we can observe that the convergence under navigational click model is slower. We attribute it to the limited feedback observed during the online interactions, because the stop probability is much higher in the navigational click model, which induce stronger position bias.

\section{Conclusion}
\label{sec:conclusion}
Existing OL2R solutions are limited to linear models, which have shown to be incompetent to capture the potential non-linear relations between queries and documents. Motivated by the recent advances in the theoretical deep learning, we propose to directly learn a neural ranker on the fly.
During the course of online learning, we assess the ranker's pairwise rank estimation uncertainty based on the tangent features of the neural network. Exploration is performed only on the pairs where the ranker is still uncertain; and for the rest of pairs we follow the predicted rank order. We prove a sub-linear upper regret bound defined on the number of mis-ordered pairs, which directly links the proposed solution's convergence with classical ranking evaluations. Our empirical experiments support our regret analysis and demonstrate significant improvement over several state-of-the-art OL2R solutions.

Our effort sheds light on deploying powerful offline learning to rank solutions online and directly optimizing rank-based metrics, e.g., RankNet and LambdaRank. Furthermore, our solution can be readily extended to more recent and advanced neural rankers (e.g., those directly learn from query-document pairs without manually constructed features). For example, the uncertainty quantification of the DNNs can be readily applied to other neural rankers (thanks to the generality of NTK and our analysis) for uncertainty-based exploration. On the other hand, computational efficiency is a practical concern for online algorithms. Our current solution requires gradient descent on the online collected training instances, which is undeniably expensive. We would like to investigate the feasibility of online stochastic gradient descent and its variants, in the setting of continual learning, which would greatly reduce the computational complexity of our solution. 

\begin{acks}
This paper is based upon the work supported by the National Science Foundation under grant IIS-1553568 and IIS-2128019, and Google Faculty Research Award.
\end{acks}

\bibliography{reference}

\clearpage
\appendix
\section{Notation}
\begin{table}[h]
  \caption{Notations used in this paper.}
  \label{tab:freq}
  \begin{tabular}{cc}
    \toprule
    Notation & Description\\
    \midrule
    $\xb_i^t$, $\xb_j^t$ & \multicolumn{1}{p{11.8cm}}{feature vector of document $i$ and $j$ under query $q_t$ at round $t$.} \\
    $m$ & \multicolumn{1}{p{11.8cm}}{the width of a DNN in each layer.} \\
    $L$ & \multicolumn{1}{p{11.8cm}}{the number of layers of a DNN.} \\
    $p$ & \multicolumn{1}{p{11.8cm}}{the total number of parameters in a DNN.} \\
    $h(\cdot), \hb$ & \multicolumn{1}{p{11.8cm}}{the underlying optimal ranking score function.} \\
    $\bgamma^*$ & \multicolumn{1}{p{11.8cm}}{optimal model for the underlying scoring function, $\bgamma^* = \btheta^* - \btheta_0$.} \\
    $\hat \bgamma_t$ & \multicolumn{1}{p{11.8cm}}{solution of the cross-entropy loss with the linearized neural network at round $t$.} \\
    $\Hb$ & \multicolumn{1}{p{11.8cm}}{the neural tangent kernel matrix for all possible query-document features.} \\
    $\lambda_0$ & \multicolumn{1}{p{11.8cm}}{$\Hb \succ \lambda_0\Ib$ minimum eigen-value of $\Hb$}\\
    $V_t, V_{\max}$ & \multicolumn{1}{p{11.8cm}}{the number of candidate documents at round $t$, and the maximum $V_t$ across all queries.} \\
    $\Psi_t$ & \multicolumn{1}{p{11.8cm}}{the set of all possible document pairs at round $t$, e.g., $\Psi_t = \{(i, j) \in \Psi_t, i \neq j\}$}\\
    $\omega_t$ & \multicolumn{1}{p{11.8cm}}{the set of certain rank orders at round $t$.} \\
    $n_t$ & \multicolumn{1}{p{11.8cm}}{the total number of query-document feature vectors until round t, which satisfies $n_t  = \sum_{s=1}^{t} V_s \leq tV_{\max}$.} \\
    $n_t^P$ & \multicolumn{1}{p{11.8cm}}{the total number of training document pairs. As we only show top-$K$ documents to the users, $n_t^P$ satisfies that $n_t^P \leq \left\lceil \frac{tK}{2} \right\rceil$} \\
    $f(\xb_i; \btheta_{t})$ & \multicolumn{1}{p{11.8cm}}{estimated ranking score of document $\xb_i$ at round $t$} \\
    $f_{ij}^t$ & \multicolumn{1}{p{11.8cm}}{the difference between the estimated ranking scores, $f_{ij}^t = f(\xb_i^t; \btheta_{t-1}) - f(\xb_j^t; \btheta_{t-1})$.} \\
    $\gb(\xb_i; \btheta_t)$ & \multicolumn{1}{p{11.8cm}}{the gradient of the neural network function at time $t$, $\gb(\xb; \btheta) = \nabla_{\btheta}f(\xb; \btheta) \in \RR^{p}$} \\
    $\gb_{ij}^t$ & \multicolumn{1}{p{11.8cm}}{the difference between the gradients at time $t$, $\gb_{ij}^t = \gb(\xb_i^t; \btheta_t) - \gb(\xb_j^t; \btheta_t)$.} \\
    $\gb_{ij}^{t,0}$ & \multicolumn{1}{p{11.8cm}}{the difference between the gradients at time $0$, $\gb_{ij}^{t,0} = \gb(\xb_i^t; \btheta_0) - \gb(\xb_j^t; \btheta_0)$.} \\
    $\eta$ & \multicolumn{1}{p{11.8cm}}{step size for gradient descent in neural network optimization.} \\
    $J$ & \multicolumn{1}{p{11.8cm}}{the number of gradient descent steps.} \\
    $\xi$ & \multicolumn{1}{p{11.8cm}}{pairwise noise in the click feedback.} \\
    $\nu$ & \multicolumn{1}{p{11.8cm}}{sub-Gaussian variable for the pairwise noise $\xi$.} \\
    $S$ & \multicolumn{1}{p{11.8cm}}{norm parameter for neural tangent kernel.} \\
    $\lambda$ & \multicolumn{1}{p{11.8cm}}{regularization parameter for loss function.} \\
    $\Ab_t$ & \multicolumn{1}{p{11.8cm}}{$\Ab_t = \sum_{s=1}^{t-1}\sum_{(i, j)\in \Omega_s^{ind}} \gb_{ij}^{s}{\gb_{ij}^{s}}^\top/m + \lambda\Ib$.} \\
    $\bar \Ab_t$ & \multicolumn{1}{p{11.8cm}}{$\bar \Ab_t = \sum_{s=1}^{t-1}\sum_{(i, j)\in \Omega_s^{ind}} \gb_{ij}^{s, 0}{\gb_{ij}^{s, 0}}^\top/m + \lambda\Ib$.} \\
    $\Delta_{\min}$ & \multicolumn{1}{p{11.8cm}}{$\Delta_{\min} = \min\nolimits_{t\in T, (i, j) \in \Psi_t}| \sigma(h_{ij}) - 1/2|$.}\\
   
  \bottomrule
\end{tabular}
\end{table}

\section{Proof of lemmas in Section 3}
\label{sec:proof3}
Before we provide the detailed proofs, we first assume that there are $n$ possible documents to be evaluated during the model learning. It is easy to conclude that $n \leq TV_{\max}$. 

First, we introduce the neural tangent kernel matrix defined on the $n$ possible query-document feature vectors across $T$ rounds, $\{\xb_i\}_{i=1}^n$.

\begin{definition}[\citet{jacot2018neural,cao2019generalization2}]\label{def:ntk}
Let $\{\xb_i\}_{i=1}^{n_T}$ be the set of all pairwise document feature vectors.  Define
\begin{align*}
    \tilde \Hb_{i,j}^{(1)} &= \bSigma_{i,j}^{(1)} = \la \xb_i, \xb_j\ra, ~~~~~~~~  \Bb_{i,j}^{(l)} = 
    \begin{pmatrix}
    \bSigma_{i,i}^{(l)} & \bSigma_{i,j}^{(l)} \\
    \bSigma_{i,j}^{(l)} & \bSigma_{j,j}^{(l)} 
    \end{pmatrix},\notag \\
    \bSigma_{i,j}^{(l+1)} &= 2\EE_{(u, v)\sim N(\zero, \Bb_{i,j}^{(l)})} \left[\phi(u)\phi(v)\right],\notag \\
    \tilde \Hb_{i,j}^{(l+1)} &= 2\tilde \Hb_{i,j}^{(l)}\EE_{(u, v)\sim N(\zero, \Bb_{i,j}^{(l)})} \left[\dot \phi(u)\dot \phi(v)\right] + \bSigma_{i,j}^{(l+1)}.\notag
\end{align*}
Then, $\Hb = (\tilde \Hb^{(L)} + \bSigma^{(L)})/2$ is called the \emph{neural tangent kernel (NTK)} matrix on the context set $\{x_i\}_{i=1}^{n_T}$, where $n_T = \sum_{s=1}^T V_s \leq TV_{\max}$ . 
\end{definition}

We also need the following assumption on the NTK matrix and the corresponding feature set.

\begin{assumption}\label{assumption:input}
$\Hb \succeq \lambda_0\Ib$; 
moreover, for any $1 \leq i \leq n$, $\|\xb_i\|_2 = 1$ and $[\xb_i]_j =[\xb_i]_{j+d/2}$.
\end{assumption}

With this assumption, the NTK matrix is assumed to be non-singular , which is mild and commonly made in literature ~\citep{du2018gradientdeep,arora2019exact, cao2019generalization2}. 
As the query-document features are manually crafted ranking features, it can be easily satisfied when no two feature vectors are in parallel. The second assumption is for  convenience in analysis and can be easily satisfied by: for any context $\xb, \|\xb\|_2 = 1$, we can construct a new context $\xb' = [\xb^\top, \xb^\top]^\top/\sqrt{2}$. Equipped with this assumption, it can be verified that with $\btheta_0$ initialized as in Algorithm~\ref{algorithm:NR}, $f(\xb_i; \btheta_0) = 0$ for any $i \in [n_T]$.

For the sigmoid function $\sigma$ applied for estimating the pairwise probability, it is well known that $\sigma$ is continuously differentiable, Lipschitz with constant $k_{\mu} = 1/4$ and $c_{\mu} = \inf \dot{\sigma} > 0$.

\subsection{Proof of Lemma~\ref{lemma:cb}}

In order to prove Lemma~\ref{lemma:cb}, we need the following technical lemmas.

\begin{lemma}[Lemma 5.1, \citet{zhou2020neural}]\label{lemma:equal}
There exists a positive constant $\bar C$ such that for any $\delta \in (0,1)$, if $m \geq \bar Cn_T^4L^6\log(n_T^2L/\delta)/\lambda_0^4$, then with probability at least $1-\delta$, there exists a $\btheta^* \in \RR^p$ such that for any $i \in [n_T]$, with $\hb = (h(\xb_1), \dots, h(\xb_n))$. 
\begin{align}
    h(\xb_i) = \la \gb(\xb_i; \btheta_0), \btheta^* - \btheta_0\ra, \quad \sqrt{m}\|\btheta^* - \btheta_0\|_2 \leq \sqrt{2\hb^\top\Hb^{-1}\hb} \leq S, \label{lemma:equal_0}
\end{align}
\end{lemma}

\begin{lemma}[Lemma B.3, \citet{zhou2020neural}]\label{lemma:newboundz}
There exist constants $\{C_i^\epsilon\}_{i=1}^5>0$ such that for any $\delta > 0$, if $m$ satisfies that
\begin{align}
     C_1^\epsilon m^{-3/2}L^{-3/2}[\log(n_TL^2/\delta)]^{3/2}\leq \tau \leq C_2^\epsilon L^{-6}[\log m]^{-3/2}\notag
\end{align}
then with probability at least $1-\delta$,
for any $t \in [T]$, we have
\begin{align}
& \|\Ab_t\|_2 \leq \lambda + C_3^\epsilon n_t^PL,\notag\\
    &\|\bar\Ab_t -  \Ab_t\|_F \leq C_4^\epsilon n_t^P \sqrt{\log(m)}\tau^{1/3}L^4,\notag\\
   & \bigg|\log \frac{\det(\bar \Ab_t)}{\det(\lambda\Ib)} - \log\frac{ \det(\Ab_t)}{\det (\lambda\Ib)}\bigg| \leq C_5^\epsilon (n_t^P)^{3/2}\lambda^{-1/2}\sqrt{\log(m)}\tau^{1/3}L^4\notag
\end{align}
For the constants, $C_3^{\epsilon} = 2C_3^z$, $C_4^{\epsilon} = 8C_3^w(C_3^z)^2$, and $C_5^{\epsilon} = 4C_3^w(C_3^z)^2$, with $C_3^z$ and $C_3^w$ from the technique lemmas, Lemma \ref{lemma:cao_gradientdifference} and Lemma~\ref{lemma:cao_boundgradient}.
\end{lemma}

\begin{lemma}[Lemma B.4, \citet{zhou2020neural}]\label{lemma:cao_functionvalue}
There exist constants $\{C_i^v\}_{i=1}^3 >0$ such that for any $\delta > 0$, if $\tau$ satisfies that
\begin{align}
     C_1^vm^{-3/2}L^{-3/2}[\log(n_TL^2/\delta)]^{3/2}\leq\tau \leq  C_2^v L^{-6}[\log m]^{-3/2},\notag
\end{align}
then with probability at least $1-\delta$,
for any $\tilde\btheta$ and $\hat\btheta$ satisfying $ \|\tilde\btheta - \btheta_0\|_2 \leq \tau, \|\hat\btheta - \btheta_0\|_2 \leq \tau$ and $i \in [n_T]$ we have
\begin{align}
    \Big|f(\xb_i; \tilde\btheta) - f(\xb_i;  \hat\btheta) - \la \gb(\xb_i; 
    \hat\btheta),\tilde\btheta - \hat\btheta\ra\Big| \leq C_3^v\tau^{4/3}L^3\sqrt{m \log m}.\notag
\end{align}
\end{lemma}

\begin{lemma}[Lemma B.5, \citet{zhou2020neural}]\label{lemma:cao_gradientdifference}
There exist constants $\{C_i^w\}_{i=1}^3>0$ such that for any $\delta \in (0,1)$, if $\tau$ satisfies that
\begin{align}
     C_1^w m^{-3/2}L^{-3/2}\max \{\log^{-3/2}m, \log^{3/2} (n_T/\delta) \}\leq\tau \leq C_2^w L^{-9/2}\log^{-3}m,\notag
\end{align}
then with probability at least $1-\delta$,
for all $\|\btheta - \btheta_0\|_2 \leq \tau$ and $i \in [n_T]$ we have
\begin{align}
  \| \gb(\xb_i; \btheta) - \gb(\xb_i; \btheta_0)\|_2 \leq C_3^w\sqrt{\log m}\tau^{1/3}L^3\|\gb(\xb_i;  \btheta_0)\|_2.\notag
\end{align}
\end{lemma}

\begin{lemma}[Lemma B.6, \citet{zhou2020neural}]\label{lemma:cao_boundgradient}
There exist constants $\{C_i^z\}_{i=1}^3>0$ such that for any $\delta > 0$, if $\tau$ satisfies that
\begin{align}
 C_1^zm^{-3/2}L^{-3/2}[\log(n_TL^2/\delta)]^{3/2}\leq\tau \leq C_2^z L^{-6}[\log m]^{-3/2},\notag
\end{align}
then with probability at least $1-\delta$,
for any $\|\btheta - \btheta_0\|_2 \leq \tau$ and $i \in[n_T]$ 
we have $\|\gb(\xb_i; \btheta)\|_F\leq C_3^z\sqrt{mL}$.
\end{lemma}

We also need the following lemmas. The first lemma is based on the generalized linear bandit~\citep{filippi2010parametric} and the analysis of linear bandit in \citep{abbasi2011improved}. For the second lemma, we adapted it from the original paper with our pairwise cross-entropy loss. The key difference lies in 1) the different number of observations in each round, which affects the required condition on the width of the neural network $m$; 2) we extend the original error bound analysis for the least square loss to the generalized linear model, e.g., logistic regression model. 

\begin{lemma}
\label{lemma:glm}
For any $t \in [T]$, with $\hat \bgamma_t$ defined as the solution of the following equation, 
\begin{align}
\label{eq:linear}
    \hat \bgamma_t = \min_{\bgamma} \sum_{s=1}^{t-1}\sum_{(i, j) \in \Omega^{ind}_s}-y_{i, j}^s\log \Big ( \sigma(\la\gb_{i, j}^{t, 0}, \bgamma\ra) \Big ) - (1 - y_{i, j}^s)\log \Big (1 -  \sigma(\la\gb_{i, j}^{t, 0}, \bgamma\ra) \Big ) + \frac{m\lambda}{2} \|\bgamma\|^2.
\end{align}
Then, with the pairwise noise $\xi_{ij}^s$satisfying Proposition~\ref{prop:pairwise}, for any $(i, j) \in \Psi_t$, with probability at least $1 - \delta_1$, we have,
\begin{align*}
    \|\sqrt{m}(\btheta^* - \btheta_0 - \hat \bgamma_t)\|_{\bar \Ab_t} 
    \leq c_{\mu}^{-2}(\sqrt{\nu^2\log(\det(\bar \Ab_t)) / (\delta_1^2\det(\lambda\Ib))} + \sqrt{\lambda}S)
\end{align*}
\end{lemma}

\begin{lemma}[Lemma B.2, \citet{zhou2020neural}]\label{lemma:newboundreference}
There exist constants $\{\bar C_i\}_{i=1}^6>0$ such that for any $\delta > 0$, if for all $t \in [T]$, $\eta$ and $m$ satisfy
\begin{align}
    &\sqrt{2n_t^P/(m\lambda)} \geq \bar C_1m^{-3/2}L^{-3/2}[\log(n_TL^2/\delta)]^{3/2},\notag \\
    &\sqrt{2n_t^P/(m\lambda)} \leq \bar C_2\min\big\{ L^{-6}[\log m]^{-3/2},\big(m(\lambda\eta)^2L^{-6}(n_t^P)^{-1}(\log m)^{-1}\big)^{3/8} \big\},\notag \\
    &\eta \leq \bar C_3(m\lambda + n_t^PmL)^{-1},\notag \\
    &m^{1/6}\geq \bar C_4\sqrt{\log m}L^{7/2}(n_t^P)^{7/6}\lambda^{-7/6}(1+\sqrt{n_t^P/\lambda}),\notag
\end{align}
then with probability at least $1-\delta$,
we have that $\|\btheta_t -\btheta_0\|_2 \leq  \sqrt{2n_t^P/(m\lambda )}$ and
\begin{align}
   \|\btheta_t - \btheta_0 - \hat \bgamma_t\|_2  \leq (1- \eta m \lambda)^{J/2} \sqrt{2n_t^P/(m\lambda)} + m^{-2/3}\sqrt{\log m}L^{7/2}(n_t^P)^{7/6}\lambda^{-7/6}(\bar C_5+ \bar C_6\sqrt{n_t^P/\lambda}).\notag
\end{align}
\end{lemma}

\begin{proof}[Proof of Lemma \ref{lemma:cb}]

We first bound the estimated pairwise preference based on the Lipschitz continuity:
\begin{align*}
    &\Big|\sigma(f(\xb_i^t; \btheta_{t-1}) - f(\xb_j^t; \btheta_{t-1})) - \sigma(h(\xb_i^t) - h(\xb_j^t))\Big|\\
    \leq& k_{\mu} \Big| f(\xb_i^t; \btheta_{t-1}) - f(\xb_j^t; \btheta_{t-1}) - \Big(h(\xb_i^t) - h(\xb_j^t)\Big)\Big|
\end{align*}

According to Lemma~\ref{lemma:equal}, and $f(\xb; \btheta_0) = 0$, we could have the following equation for document $\xb_i^t$,
\begin{align*}
    f(\xb_i^t; \btheta_{t-1}) - h(\xb_i^t) =& f(\xb_i^t; \btheta_{t-1}) - f(\xb_i^t; \btheta_0) - \la \gb(\xb_i^t; \btheta_{t-1}), \btheta_{t-1} - \btheta_0\ra  \\
    & + \la \gb(\xb_i^t; \btheta_{t-1}), \btheta_{t-1} - \btheta_0\ra - \la \gb(\xb_i^t; \btheta_{t-1}), \btheta^* - \btheta_0\ra \\
    & + \la \gb(\xb_i^t; \btheta_{t-1}), \btheta^* - \btheta_0\ra - \la \gb(\xb_i^t; \btheta_0), \btheta^* - \btheta_0\ra.
\end{align*}

Therefore, we could have the following inequalities based on the triangle inequality.
\begin{align*}
    &\Big| f(\xb_i^t; \btheta_{t-1}) - f(\xb_j^t; \btheta_{t-1}) - \Big(h(\xb_i^t) - h(\xb_j^t)\Big)\Big| \\
    \leq & \Big| \la\gb(\xb_i^t; \btheta_{t-1}) - \gb(\xb_j^t; \btheta_{t-1}), \btheta_{t-1} - \btheta^* \ra\Big| \\
    & + \|\btheta^* - \btheta_0\|_2\Big( \|\gb(\xb_i^t;\btheta_{t-1}) - \gb(\xb_i^t;\btheta_0)\|_2 + \|\gb(\xb_j^t;\btheta_{t-1}) - \gb(\xb_j^t;\btheta_0)\|_2\Big) \\ 
    &+ \Big|f(\xb_i^t; \btheta_{t-1}) - f(\xb_i^t; \btheta_0) - \la \gb(\xb_i^t; \btheta_{t-1}), \btheta_{t-1} - \btheta_0\ra\Big| \\
    &+ \Big|f(\xb_j^t; \btheta_{t-1}) - f(\xb_j^t; \btheta_0) - \la \gb(\xb_j^t; \btheta_{t-1}), \btheta_{t-1} - \btheta_0\ra\Big| \\
    \leq & 2C_3^v\tau^{4/3}L^3\sqrt{m \log m} + 2C_3^zC_3^w\tau^{1/3}L^{7/2}\sqrt{\log(m)}S + \Big| \la\gb_{ij}^t, \btheta_{t-1} - \btheta^* \ra\Big|,
\end{align*}
where the last inequality is due to Lemma~\ref{lemma:equal}, \ref{lemma:cao_functionvalue}, \ref{lemma:cao_gradientdifference}, \ref{lemma:cao_boundgradient}, with satisfied $\tau$ as the upper bound of $\|\btheta - \btheta_0\|_2$.

Now we start to bound the last term $\Big| \la\gb_{ij}^t, \btheta_{t-1} - \btheta^* \ra\Big|$.
\begin{align}
    \Big| \la\gb_{ij}^t, \btheta_{t-1} - \btheta^* \ra\Big| =& \Big| \la\gb_{ij}^t, \btheta_{t-1} -\btheta_0 - \hat\bgamma_t - (\btheta^* -\btheta_0 - \hat\bgamma_t) \ra\Big| \notag\\
    \leq& |\la\gb_{ij}^t, \btheta^* -\btheta_0 - \hat\bgamma_t\ra| + \|\gb_{ij}^t\|\|\btheta_{t-1} -\btheta_0 - \hat\bgamma_t\| \label{eq:2}
\end{align}
For the first term, we have the following analysis.
\begin{align*}
    |\la\gb_{ij}^t, \btheta^* -\btheta_0 - \hat\bgamma_t\ra| &\leq \|\gb_{ij}^t/\sqrt{m}\|_{\Ab_t^{-1}}\|\sqrt{m}(\btheta^* -\btheta_0 - \hat\bgamma_t)\|_{\Ab_t} \\
    &\leq \|\gb_{ij}^t/\sqrt{m}\|_{\Ab_t^{-1}} \sqrt{(1 +\|\Ab_t - \bar\Ab_t\|_2/\lambda)}\|\sqrt{m}(\btheta^* -\btheta_0 - \hat\bgamma_t)\|_{\bar\Ab_t} \\ 
    &\leq \sqrt{1 + C_4^\epsilon n_t^P \sqrt{\log(m)}\tau^{1/3}L^4}\|\sqrt{m}(\btheta^* -\btheta_0 - \hat\bgamma_t)\|_{\bar\Ab_t} \|\gb_{ij}^t/\sqrt{m}\|_{\Ab_t^{-1}},
\end{align*}
where the first inequality is trivial, and the second inequality is due to the fact that $\xb^\top\Pb\xb \leq \xb^\top\Qb\xb \cdot \|\Pb\|_2/\lambda_{\min}(\Qb)$, and $\lambda_{\min}(\bar \Ab_t) \geq \lambda$, the third inequality is based on Lemma~\ref{lemma:newboundz} with $\|\Ab_t - \bar\Ab_t\|_2 \leq \|\Ab_t - \bar\Ab_t\|_F$. According to Lemma~\ref{lemma:glm}, with probability $1 - \delta_1$, we have 
\begin{align*}
    \|\sqrt{m}(\btheta^* - \btheta_0 - \hat \bgamma_t)\|_{\bar \Ab_t} \leq& c_{\mu}^{-2}(\sqrt{\nu^2\log(\det(\bar \Ab_t)) / (\delta_1^2\det(\lambda\Ib))} + \sqrt{\lambda}S) \\
    \leq & c_{\mu}^{-2}(\sqrt{\nu^2\log(\det(\Ab_t)) / (\delta_1^2\det(\lambda\Ib)) + C_5^\epsilon (n_t^P)^{3/2}\lambda^{-1/2}\sqrt{\log(m)}\tau^{1/3}L^4} + \sqrt{\lambda}S)
\end{align*}
where the second inequality is based on Lemma~\ref{lemma:newboundz}.
For the second term of Eq~\ref{eq:2}, it can be bounded according to Lemma~\ref{lemma:cao_boundgradient} and Lemma~\ref{lemma:newboundreference}.
By chaining all the inequalities, and with $\|\btheta - \btheta_0\| \leq \tau \leq \sqrt{2n_t^P/(m\lambda)}$, and the satisfied $m$ and $\eta$, we complete the proof.
\end{proof}

\section{Proofs of lemmas and theorems in Section 4}
\label{sec:proof4}

Before we provide the detailed proofs, we need the following technique lemmas.
\begin{lemma}
\label{lemma:bound}
Let $a$ and $b$ be two positive constants, if $m \geq a^2 + 2b$, then $m - a\sqrt{m} - b \geq 0$
\end{lemma}
The following lemma is derived from random matrix theory. We adapted it from Equation (5.23) of Theorem 5.39 from \citep{vershynin2010introduction}.
\begin{lemma}
\label{lemma:matrix}
Let $\Mb \in \mathbb{R}^{N \times p}$ be a matrix whose rows $\Mb_i$ are independent sub-Gaussian isotropic random vectors in $\RR^p$ with parameter $\rho$, namely $\EE[\text{exp}(\gb_{\ijp}^\top (\Mb_i - \EE[\Mb_i])/\sqrt{m}] \leq \text{exp}(\rho^2\Vert \gb_{\ijp}/\sqrt{m} \Vert ^2 / 2)$ for any $\gb_{\ijp} \in \RR^p$. Then, there exist positive universal constants $C_1$ and $C_2$ such that, for every $t \geq 0$, the following holds with probability at least $1 - 2\text{exp}(-C_2t^2), \text{where } \upsilon = \rho(C_1\sqrt{p/N} + t/\sqrt{N})$: $\Vert\frac{1}{N}\Mb^\top \Mb - \mathbf{I}_p\Vert \leq \text{max}\{\upsilon, \upsilon^2\}$.
\end{lemma}

\subsection{Proof of Lemma~\ref{lemma:uncertain}}
\begin{proof}[Proof of Lemma~\ref{lemma:uncertain}]
In this proof we will first provide an analysis on the minimum eigenvalue of $\bar\Ab_t$, and then provide the detailed derivation of the upper bound of the probability.

At initialization, DNNs are equivalent to Gaussian processes in the infinite-width limit. Thus, we assume that the gradient differences between the documents at the initial step are random vectors drawn from some distribution $v$. With $\bSigma = \EE[\gb_{ij}^0{\gb_{ij}^0}^\top]$ as the second moment matrix, define $\Zb = \bSigma^{-1/2}\Xb$, where $\Xb$ is a random vector drawn from the same distribution $v$. Then $\Zb$ is isotropic, namely $\mathbb{E}[\Zb\Zb^\top] = \Ib_p$. Define $\Db = \sum_{s=1}^{t-1}\sum_{(\ip, \jp) \in \Omega_s^{ind}} \Zb_{\ijp}^s{\Zb_{\ijp}^{s\top}}$, where $\Zb_{\ijp}^s = \Sigma^{-1/2}\gb_{\ijp}^{s ,0}$. It is trivial to have $\Db = \Sigma^{-1/2}(\bar \Ab_t  - \lambda\Ib)\Sigma^{-1/2}$.  From Lemma~\ref{lemma:matrix}, we know that for any $l$, with probability at least $1 - 2\text{exp}(-C_2l^2)$,
$\lambda_{\text{min}}(\Db) \geq n_t - C_1\sigma^2n_t - \sigma^2l\sqrt{n_t}$,
where $\sigma$ is the sub-Gaussian parameter of $\Zb$, which is upper-bounded by  $\Vert\bSigma^{-1/2}\Vert = \lambda_{\text{min}}(\bSigma)$, and $n_t = \sum_{s=1}^{t-1}|\Omega_s^{ind}|$, represents the number of pairwise observations so far. Thus, we can rewrite the above inequality which holds with probability $1 - \delta_2$ as
$\lambda_{\text{min}}(\Db) \geq n_t - \lambda_{\text{min}}^{-1}(\bSigma)(C_1n_t + l\sqrt{n_t})$, and:
\begin{align*}
    \lambda_{\text{min}}(\bar \Ab_t - \lambda\Ib) &= \min_{x\in \mathbb{B}^p}x^\top  (\bar \Ab_t - \lambda\Ib) x = \min_{x\in \mathbb{B}^p}x^\top\bSigma^{1/2}\Db\bSigma^{1/2}x \nonumber \\
    &\geq \lambda_{\text{min}}(\Db)\min_{x\in \mathbb{B}^p}x^\top \Sigma x = \lambda_{\text{min}}(\Db)\lambda_{\text{min}}(\bSigma)\\
    &\geq \lambda_{\text{min}}(\bSigma)n_t - C_1n_t - C_2\sqrt{n_t\log(1/\delta_2)}
\end{align*}

Under event $E_t$, based on the definition of $\omega_t$ in Section \ref{sec:method}, we know that for any document $i$ and $j$ at round $t$, $(i, j) \notin \omega_t$ if and only if $\sigma(f_{ij}^t) - \alpha_t\Vert\gb_{ij}^{t}/\sqrt{m}\Vert_{\Ab_t^{-1}}  - \epsilon(m) \leq 1/2$ and $\sigma(f_{ji}^t) - \alpha_t\Vert\gb_{ji}^{t}/\sqrt{m}\Vert_{\Ab_t^{-1}} -\epsilon(m) \leq 1/2$. 

For a logistic function, we know that $\sigma(s) = 1 - \sigma(-s)$. Therefore, according to Lemma~\ref{lemma:cb}, let $CB_{ij}^t$ denote $\alpha_t\Vert\gb_{ij}^t/\sqrt{m}\Vert_{\Ab_t^{-1}} + \epsilon(m)$, we can conclude that $(i, j) \notin \omega_t$ if and only if $|\sigma(f_{ij}^t)- 1/2| \leq CB_{ij}^t$; and accordingly, $(i, j) \in \omega_t$, when $|\sigma(f_{ij}^t) - 1/2| > CB_{ij}^t$.  

According to the discussion above, at round $t$, the probability that the estimated preference between document $i$ and $j$ to be in an uncertain rank order, i.e., $(i, j) \notin \omega_t$, can be upper bounded by:
\begin{align*}
    & \mathbb{P}\big((i, j) \notin \omega_t\big) = \mathbb{P}\big(|\sigma(f_{ij}^t) - 1/2| \leq CB_{ij}^t\big) \\
    \leq& \mathbb{P}\left(\left||\sigma(f_{ij}^t) - \sigma(h_{ij}^t)| - |\sigma(h_{ij}^t) - 1/2|\right| \leq CB_{ij}^t\right)  \\
    \leq& \mathbb{P} \left(|\sigma(h_{ij}^t) - 1/2| - |\sigma(f_{ij}^t) - \sigma(h_{ij}^t)| \leq CB_{ij}^t\right) \\
    \leq & \mathbb{P} \left(\Delta_{\min} - |\sigma(f_{ij}^t) - \sigma(h_{ij}^t)| \leq CB_{ij}^t\right), 
\end{align*}
where the first inequality is based on the reverse triangle inequality.
The last inequality is based on the definition of $\Delta_{\min}$. Based on Lemma~\ref{lemma:cb}, the above probability can be further bounded by
\begin{align*}
    &\mathbb{P} \left(\Delta_{\min} - |\sigma(f_{ij}^t) - \sigma(h_{ij}^t)| \leq CB_{ij}^t\right) \\
    =& \mathbb{P}\left(|\sigma(f_{ij}^t) - \sigma(h_{ij}^t)| \geq \Delta_{\min} - \alpha_t\Vert\gb_{ij}^t/\sqrt{m}\Vert_{\sub \Ab_t^{-1}} - \epsilon(m)\right) \\
    \leq& \mathbb{P} \left(\frac{2k_{\mu}}{c_{\mu}} ||\gb_{ij}^{t}/\sqrt{m}||_{\sub \Ab_t^{-1}}
    \left(\left\| \Wb_t\right\|_{\sub \Ab_t^{-1}} + \sqrt{\lambda}S\right) \geq \Delta_{\min} - \alpha_t\Vert\gb_{ij}^{t}/\sqrt{m}\Vert_{\sub \Ab_t^{-1}}- 2\epsilon(m)\right)  \\
    \leq& \mathbb{P}\left(\left\| \Wb_t\right\|_{\sub \Ab_t^{-1}} \geq \frac{c_{\mu} (\Delta_{\min} - 2\epsilon(m))}{2k_{\mu}||\gb_{ij}^{t}/\sqrt{m}||_{\sub \Ab_{t}^{-1}}} - \left(\sqrt{\nu^2\log{\frac{\det(\Ab_t)}{\delta_1^2 \det(\lambda \mathbf{I})}}} + 2\sqrt{\lambda} S\right)\right).
\end{align*}
where $\Wb_t = \sum_{s = 1}^t \sum_{(\ip, \jp) \in \Omega_s^{ind}} \xi_{\ijp}^s\gb_{\ijp}^{s}$.

For the right-hand side, we know that $\lambda_{\min}(\Ab_t) \geq \lambda_{\min}(\bar \Ab_t) + \|\Ab_t - \bar\Ab_t\| \geq  \lambda_{\min}(\bar \Ab_t -\lambda\Ib) + \lambda + \|\Ab_t - \bar\Ab_t\|$.
With some positive constants $\{C_i^u\}_{i=1}^5$, for $t \geq t^\prime = \big(C_1^u + C_2^u\sqrt{\log(1/\delta_2)} + C_3^u V_{\max}\big)^2 + C_4^u\log(1/\delta_1) + C_5^u$, as $n_t > t$, we have $n_t - \sqrt{n_t}\big(C_1^u + C_2^u\sqrt{\log(1/\delta_2)} + C_3^u V_{\max}\big) > C_4^u\log(1/\delta_1) + C_5^u$. Hence, we have the following inequalities,
\begin{align*}
&\left(\frac{c_{\mu}(\Delta_{\min} - 2\epsilon(m))}{2k_{\mu}||\gb_{ij}^t/\sqrt{m}||_{\sub \Ab_{t}^{-1}}}
\right)^2 - \left(\sqrt{\nu^2\log{\frac{\det(\Ab_t)}{\delta_1^2 \det(\lambda \mathbf{I})}}} + 2\sqrt{\lambda} S\right)^2 \\
\geq& {\lambda_{\min}(\Ab_t)}c_{\mu}^2(\Delta_{\min} - 2\epsilon(m))^2/(4C_3^zk_{\mu}^2L)  - \nu^2\log({\det(\Ab_t)}/{(\delta_1^2\det(\lambda\mathbf{I}))}) \\
& - 4\lambda S^2 - 4\sqrt{\lambda}S\nu\sqrt{\log({\det(\Ab_t)}/{\det{\lambda\mathbf{I}}}) + \log({1}/{\delta_1^2})} \\
\geq& (\lambda_{\min}(\bar \Ab_t -\lambda\Ib) + \lambda + \|\Ab_t - \bar\Ab_t\|)c_{\mu}^2(\Delta_{\min} - 2\epsilon(m))^2/(4C_3^zk_{\mu}^2L) \\
&- (4\sqrt{\lambda}S\nu + \nu^2)\left(\log({\det(\Ab_t)}/{\det(\lambda\Ib)})  + \log({1}/{\delta_1^2})\right)- 4\lambda S^2\\
\geq & \lambda_{\min}(\bSigma)(n_t - \sqrt{n_t}\big(C_1^u + C_2^u\sqrt{\log(1/\delta_2)} + C_3^uV_{\max}\big) - C_4^u\log(1/\delta_1) + C_5^u) \geq  0
\end{align*}
with corresponding positive constants $\{C_i^u\}_{i=1}^5$.
Therefore, the probability could be upper bounded:
\begin{align*}
    &\PP \left(\Delta_{\min} - |\sigma(f_{ij}^t) - \sigma(h_{ij}^t)| \leq CB_{ij}^t\right) \\\leq
    &\mathbb{P}\left(\left\| \Wb_t\right\|_{\sub \Ab_t^{-1}}^2  {\geq} \left(\frac{c_{\mu} (\Delta_{\min} - 2\epsilon(m))}{2k_{\mu}\Vert\gb_{ij}^{t}/\sqrt{m}\Vert_{\sub \Ab_{t}^{-1}}} - \Big(\sqrt{\nu^2\log{\frac{\det(\Ab_t)}{\delta_1 \det(\lambda \mathbf{I})}}} + 2\sqrt{\lambda} S\Big)\right)^2\right)\\
    \leq &  \mathbb{P}\left(\left\| \Wb_t\right\|_{\sub \Ab_t^{-1}}^2  {\geq}  \frac{ c_{\mu}^2(\Delta_{\min} - 2\epsilon(m))^2}{4k_{\mu}^2\Vert\gb_{ij}^{t}/\sqrt{m}\Vert_{\Ab_t}^{-1}} + \nu^2\log(\frac{\det(\Ab_t)}{\delta_1^2\det(\lambda\mathbf{I})})\right)\\
    \leq & \mathbb{P}\left(\left\| \Wb_t\right\|_{\sub \Ab_t^{-1}}^2  {\geq}  2\nu^2\log\left(\exp\left(\frac{c_{\mu}^2(\Delta_{\min} - 2\epsilon(m))^2}{8\nu^2k_{\mu}^2\Vert\gb_{ij}^{t}/\sqrt{m}\Vert_{\Ab_t}^{-1}}\right) \cdot \frac{\det(\Ab_t)}{\delta_1^2\det(\lambda\mathbf{I})}\right)\right) \\
    \leq & \delta_1 \cdot \exp^{-1}\left(\frac{c_{\mu}^2(\Delta_{\min} - 2\epsilon(m))^2}{8\nu^2k_{\mu}^2\Vert\gb_{ij}^{t}/\sqrt{m}\Vert_{\Ab_t^{-1}}^2}\right) \leq  C_6^u\log(1/\delta_1)\frac{\|\gb_{ij}^{t}/\sqrt{m}\|_{\Ab_t^{-1}}^2}{(\Delta_{\min} - 2\epsilon(m))^2},
\end{align*}
with an additional positive constant $C_6^{^u}$. This completes the proof.
\end{proof}

\subsection{Proof of Theorem~\ref{thm:upper-regret}}

\begin{lemma}\label{lemma:newboundregret}
There exist positive constants $\{C_i\}_{i=1}^2$ such that for any $\delta \in (0,1)$, if $\eta \leq \bar C_1(TmL + m\lambda)^{-1}$ and $m \geq \bar C_2\max\big\{T^7\lambda^{-7}L^{21}(\log m)^3, N^6L^6(\log(n_TL^2/\delta))^{3/2}\big\}$, 
then with probability at least $1-\delta$, we have
\begin{align*}
     \sum\nolimits_{t=1}^T\sum\nolimits_{(\ip, \jp)\in \Omega_t} \|\gb_{\ijp}^{t}/\sqrt{m}\|_{\Ab_t^{-1}} \leq 2 \log \frac{\det \Ab_T}{\det \lambda \Ib} \leq \tilde d\log(1 + TV_{\max}^2/\lambda) + 1
\end{align*}
where $\tilde d$ is defined as the effective dimension of $\Hb$.
\end{lemma}

\begin{proof}[Proof of Theorem~\ref{thm:upper-regret}]
With $\delta_1$ and $\delta_2$ defined in the previous lemmas, we have with probability at least $1 - \delta_1$, the $T$-step regret is upper bounded as:
\begin{align}
   R_T = R_{t^\prime} + R_{T - t^\prime} \leq t^\prime * V_{\max}^2 + (T - t^\prime)\delta_2 V_{\max^2} + (1 - \delta_2)\sum\nolimits_{t=t^\prime}^T r_t
\label{eq:regret_all}
\end{align}
When event $E_t$ and the event defined in Lemma~\ref{lemma:cb} both occur, the instantaneous regret at round $t$ is bounded by $r_t = \mathbb{E} \big[K(\tau_s, \tau_s^*)\big] \leq \EE[U_t]$, where $U_t$ denotes the number of uncertain rank orders under the ranker at round $t$. As the ranked list is generated by topological sort on the certain rank orders, the random shuffling only happens between the documents that are in uncertain rank orders, which induce regret in the proposed ranked list. In each round of result serving, as the model $\btheta_t$ would not change until the next round, the expected number of uncertain rank orders can be estimated by summing the uncertain probabilities over all possible pairwise comparisons under the current query $q_t$, e.g., $\EE[U_t] = 1/2 \sum_{(i, j) \in \Psi_t}  \mathbb{P}((i, j) \notin \omega_t)$.

Based on Lemma~\ref{lemma:uncertain}, the cumulative number of mis-ordered pairs can be bounded by the probability of observing uncertain rank orders in each round, which shrinks with more observations become available over time, 
\begin{align*}
    \EE\big[\sum\nolimits_{s=t^\prime}^{T} U_t\big] \leq&  \EE\big[1/2\sum\nolimits_{s=t^\prime}^{T} \sum\nolimits_{(\ip, \jp) \in \Psi_s} \PP((\ip, \jp) \notin \omega_t)\big] \\
    \leq& \EE\big[1/2\sum\nolimits_{s=t^\prime}^{T} \sum\nolimits_{(\ip, k^\prime) \in \Psi_s} C_6^{u}\log(1/\delta_1){\|\gb_{\ijp}^{t}\|_{\Ab_t^{-1}}^2}/{(\Delta_{\min} - 2\epsilon(m))^2}\big].
\end{align*}

Because $\Ab_t$ only contains information of observed document pairs so far, our algorithm guarantees the number of mis-ordered pairs among the observed documents in the above inequality is upper bounded. 
To reason about the number of mis-ordered pairs in those unobserved documents (i.e., from $o_t$ to $L_t$ for each query $q_t$), we leverage the constant $p^*$, which is defined as the minimal probability that all documents in a query are examined over time, 
\begin{align*}
    &\EE\big[\sum\nolimits_{t=\tp}\sum\nolimits_{(\ip, \jp)\in\Psi_t} \|\gb_{\ijp}^{t}/\sqrt{m}\|_{\Ab_{t}^{-1}}\big] \\ =& \EE\big[\sum\nolimits_{t=\tp}\sum\nolimits_{(\ip, \jp)\in\Psi_t} \|\gb_{\ijp}^{t}/\sqrt{m}\|_{\Ab_{t}^{-1}} \times \EE\big[{p_t^{-1}}\textbf{1}\{o_t = V_t\}\big]\big] \\
    \leq & p^*{^{-1}}\EE\big[\sum\nolimits_{t=\tp}\sum\nolimits_{(\ip, \jp)\in\Psi_t} \|\gb_{\ijp}^{t}/\sqrt{m}\|_{\Ab_{t}^{-1}}\textbf{1}\{o_t = V_t\}\big]
\end{align*}

Besides, we only use the independent pairs, $\Omega_t^{ind}$ to update the model and the corresponding $\Ab_t$ matrix. Therefore, to bound the regret, we rewrite the above equation as:
\begin{align}
\label{eq:9}
    &\mathbb{E}\left[\sum\nolimits_{t=\tp}^T\sum\nolimits_{(\ip, \jp) \in \Psi_t} \Vert\gb_{\ijp}^{t}/\sqrt{m}\Vert_{\Ab_t^{-1}}^2\right]  \\
    =& \mathbb{E}\left[\sum\nolimits_{t=\tp}^T\sum\nolimits_{(\ip, \jp)  \in \Omega_t^{ind}} \left (\Vert\gb_{\ijp}^{t}/\sqrt{m}\Vert_{\Ab_t^{-1}}^2 + \sum\nolimits_{k\in [V_t] \setminus \{\ip, \jp\} } \Vert\gb_{\ip k}^{t}/\sqrt{m}\Vert_{\Ab_t^{-1}}^2 + \Vert\gb_{\jp k}^{t}/\sqrt{m}\Vert_{\Ab_t^{-1}}^2\right)\right] \nonumber \\
    =& \mathbb{E}\left[\sum\nolimits_{t=\tp}^T\sum\nolimits_{(\ip, \jp)  \in \Omega_t^{ind}} \left ((L_t - 1)\Vert\gb_{\ijp}^{t}/\sqrt{m}\Vert_{\Ab_s^{-1}}^2 + \sum\nolimits_{k\in [V_t] \setminus \{\ip, \jp\} } (2/m){\gb_{\ip k}^{t}}^\top\Ab_t^{-1}\gb_{\jp k}^{t}\right)\right] \nonumber
\end{align}

For the second term, it can be bounded as:
\begin{align*}
    \sum\nolimits_{t=\tp}^T\sum\nolimits_{(\ip, \jp)  \in \Omega_t^{ind}}\sum\nolimits_{k\in [V_t] \setminus \{\ip, \jp\} } (2/m){\gb_{\ip k}^{t}}^\top\Ab_t^{-1}\gb_{\jp, k}^{t} 
    \leq \sum\nolimits_{t=\tp}^T {2C_3^z(V_{\max}^2 - 2V_{\max})L^2 }/{\lambda_{\min}(\Ab_t)}
\end{align*}
where the first inequality is due to Lemma~\ref{lemma:cao_boundgradient}. According to the analysis of $\lambda_{\min}(\Ab_t)$ and $\lambda_{\min}(\bar \Ab_t)$, the convergence rate the above upper bound is faster than the self-normalized term in Eq~\ref{eq:9}. Hence, by chaining all the inequalities, we have with probability at least $ 1- \delta_1$, the regret satisfies,
\begin{align*}
    R_T  \leq&  R^\prime + (1 - \delta_2)C_6^u \log(1/\delta_1) (w + V_{\max}(\tilde d \log (1 + TV_{\max}^2/\lambda)  + 1)/{(\Delta_{\min} - 2\epsilon(m))^2}\\
    \leq &  R^\prime + (C^r_1\log(1/\delta_1)\tilde{d}\log(1 + TV_{\max}^2/\lambda) + C^r_2)(1 - \delta_2)/p^*
\end{align*}
where $\{C_i^r\}_{i=1}^2$ are positive constants, $R^{\prime} = t^{\prime}V_{\max}^2 + (T - t^{\prime})\delta_2V_{\max}^2$.
By choosing $\delta_1 = \delta_2 = 1/T$, the theorem shows that the expected regret is at most $R_T \leq O(\log^2(T))$. 
\end{proof}

\section{Proofs of lemmas in Appendix \ref{sec:proof3}}
In this section, we provide the detailed proofs of Lemma~\ref{lemma:newboundreference} and Lemma~\ref{lemma:glm} in Section B. For the technical lemmas, interested readers can refer to the original paper to \citep{zhou2020neural} for more details.

We need the following technical lemma adopted from \citep{zhou2020neural}.
\begin{lemma}[Lemma 5.1, \citet{zhou2020neural}]\label{lemma:ntkconverge}
Let $\Gb = [\gb(\xb_1; \btheta_0),\ldots,\gb(\xb_{n_T}; \btheta_0)]/\sqrt{m} \in \RR^{p \times n_T}$. Let $\Hb$ be the NTK matrix as defined in Definition \ref{def:ntk}. For any $\delta \in(0,1)$, if
\begin{align}
    m  = \Omega\bigg(\frac{L^6\log (n_TL/\delta)}{\epsilon^4}\bigg),\notag
\end{align}
then with probability at least $1-\delta$,
we have $\|\Gb^\top\Gb - \Hb\|_F \leq n\epsilon$.
\end{lemma}


\subsection{Proof of Lemma~\ref{lemma:newboundreference}}
\label{proof:newboundreference}
In this section, we will provide the detailed proof of Lemma~\ref{lemma:newboundreference}. First, assume that until round $t$, there are in total $n_t$ observed document pairs, e.g., $\sum_{s=1}^{t-1}|\Omega_s^{ind}| = n_t \leq V_{\max}t$, where $|\cdot|$ represents the cardinality of the designated set, and $V_{\max}$ is the maximum number of document pairs that can be observed given query $q$ across all queries. For simplicity, we will re-index all the observed pairs until round $t$ from $1$ to $n_t$ in the following analysis.

Then, for round $t$, define the following quantities,
\begin{align}
    &\Jb^{(j)} = \Big(\gb(\xb_{1, 1}; \btheta^{(j)}) - \gb(\xb_{1, 2}; \btheta^{(j)}, \dots, \gb(\xb_{n_t, 1}; \btheta^{(j)}) - \gb(\xb_{n_t, 2}; \btheta^{(j)})\Big) \in \RR^{p \times n_t} \\ 
    &\Hb^{(j)} = \Jb^{(j)\top} \Jb^{(j)} \in \RR^{n_t \times n_t} \\
    & \fb^{(j)} = \Big(f(\xb_{1, 1}; \btheta^{(j)}) - f(\xb_{1, 2}; \btheta^{(j)}), \dots, f(\xb_{n_t, 1}; \btheta^{(j)}) - f(\xb_{n_t, 2}; \btheta^{(j)})\Big)^\top \in \RR^{n_t \times 1} \\
    & \pb^{(j)} = \Big(\sigma(f(\xb_{1, 1}; \btheta^{(j)}) - f(\xb_{1, 2}; \btheta^{(j)})), \dots, \sigma(f(\xb_{n_t, 1}; \btheta^{(j)}) - f(\xb_{n_t, 2}; \btheta^{(j)})\Big)^\top \in \RR^{n_t \times 1} \\
    &\yb = \big(y_{1}, \dots y_{n_t}\big)^\top \in \RR^{n_t \times 1}
\end{align}
According to the loss function defined in Eq~\eqref{eq:loss}, we have the update rule of $\btheta^{(j)}$ as follows:
\begin{align}
    \btheta^{(j+1)} = \bthetaj - \eta[\Jb^{(j)}(\pb^{(j)} - \yb) + m\lambda(\bthetaj -\bthetaz)]
\end{align}
Besides, we have the following auxiliary sequence $\{\tilde \btheta^{(k)}\}$,
\begin{align*}
    \tbthetaz = \bthetaz, \tilde \btheta^{(j+1)} = \tbthetaj - \eta[\Jb^{(0)}(\sigma(\Jb^{(0)\top}(\tbthetaj - \tbthetaz)) - \yb) + m\lambda(\tbthetaj - \tbthetaz)].
\end{align*}

Next lemma provides perturbation bounds for $\Jb^{(j)}, \Hb^{(j)}$ and $\|\fb^{(j+1)} - \fb^{(j)} - [\Jb^{(j)}]^\top(\btheta^{(j+1)} - \btheta^{(j)})\|_2$.
\begin{lemma}\label{lemma:jacob}
There exist constants $\{\hat C_i\}_{i=1}^6>0$ such that for any $\delta > 0$, if $\tau$ satisfies that
\begin{align}
     \hat C_1m^{-3/2}L^{-3/2}[\log(n_TL^2/\delta)]^{3/2}\leq\tau \leq \hat C_2 L^{-6}[\log m]^{-3/2},\notag
\end{align}
then with probability at least $1-\delta$,
for any $j, s \in [J]$, $\|\btheta^{(j)} -\btheta^{(0)}\|_2\leq \tau$ and $\|\btheta^{(s)} -\btheta^{(0)}\|_2\leq \tau$, we have the following inequalities,
\begin{align}
     &\big\|\Jb^{(j)}\big\|_F \leq \hat C_4 \sqrt{n_t^PmL},\label{eq:jacobcc1}\\
     &\|\Jb^{(j)} - \Jb^{(0)}\|_F  \leq \hat C_5\sqrt{n_t^Pm\log m}\tau^{1/3}L^{7/2},\label{eq:jacobcc2}\\
&\big\|\fb^{(s)} - \fb^{(j)} - [\Jb^{(j)}]^\top(\btheta^{(s)} - \btheta^{(j)})\big\|_2 \leq 
    \hat C_6 \tau^{4/3}L^3 \sqrt{n_t^Pm \log m}, \label{eq:jacobcc3} \\
    &\|\yb\|_2 \leq \sqrt{n_t^P}.\label{eq:jacobcc4}
\end{align}
\end{lemma}

\begin{lemma}\label{lemma:linearconvergence}
There exist constants $\{\tilde C_i\}_{i=1}^4>0$ such that for any $\delta > 0$, if $\tau, \eta$ satisfy that
    \begin{align}
    &\tilde C_1m^{-3/2}L^{-3/2}[\log(n_TL^2/\delta)]^{3/2}\leq\tau \leq \tilde C_2 L^{-6}[\log m]^{-3/2},\notag ,\\
    &\eta \leq \tilde C_3(m\lambda + n_t^PmL)^{-1}, \tau^{8/3} \leq \tilde C_4m(\lambda\eta)^2L^{-6}(n_t^P)^{-1}(\log m)^{-1},\notag
\end{align}
then with probability at least $1-\delta$,
for any $j \in [J]$, $\|\btheta^{(j)}-\btheta^{(0)}\|_2\leq \tau$,  we have $\|\pb^{(j)} - \yb\|_2\leq 2\sqrt{n_t^P}$.
\end{lemma}

Next lemma gives an upper bound of the distance between auxiliary sequence $\|\tilde \btheta^{(j)} -  \btheta^{(0)}\|_2$.
\begin{lemma}\label{lemma:implicitbias}
There exist constants $\{C_i\}_{i=1}^3>0$ such that for any $\delta \in (0,1)$, if $\tau, \eta$ satisfy that
\begin{align}
    & C_1m^{-3/2}L^{-3/2}[\log(n_t^PL^2/\delta)]^{3/2}\leq\tau \leq C_2 L^{-6}[\log m]^{-3/2}, \eta \leq  C_3(n_t^PmL + m\lambda)^{-1} ,\notag
\end{align}
then with probability at least $1-\delta$, we have that for any $j \in[J]$,
\begin{align*}
    \big\|\tilde\btheta^{(j)} - \btheta^{(0)}\big\|_2 \leq \sqrt{2n_t^P/(m\lambda)}, \text{and } \big\|\tilde\btheta^{(j)} - \btheta^{(0)} - \hat \wb_t\big\|_2 \leq (1 - \eta m \lambda)^{j/2} \sqrt{n_t^P/(m\lambda)}.
\end{align*}
\end{lemma}

With above lemmas, we prove Lemma \ref{lemma:newboundreference} as follows.
\begin{proof}[Proof of Lemma \ref{lemma:newboundreference}]
Set $\tau = \sqrt{2n_t^P/(m\lambda)}$. First we assume that $\|\btheta^{(j)}-\btheta^{(0)}\|_2\leq\tau$ for all $0 \leq j \leq J$. Then with this assumption and the choice of $m,\tau$, we have that Lemma \ref{lemma:jacob},  \ref{lemma:linearconvergence} and \ref{lemma:implicitbias} hold. Then we have
\begin{align*}
&\big\|\btheta^{(j+1)} - \tilde \btheta^{(j+1)}\big\|_2  \\
=& \big\|\bthetaj - \tbthetaj - \eta \Jb^{(j)}(\pb^{(j)} - \yb) - \eta m\lambda(\bthetaj - \bthetaz) + \eta\Jb^{(0)}(\sigma(\Jb^{(0)\top}(\tbthetaj - \tbthetaz)) - \yb) + \eta m \lambda(\tbthetaj - \tbthetaz)\big\|_2 \\
= & \big\|(1 - \eta m\lambda)(\bthetaj - \tbthetaj) - \eta(\Jb^{(j)} - \Jb^{(0)})(\pb^{(j)} - \yb) - \eta \Jb^{(0)}(\pb^{(j)}- \sigma(\Jb^{(0)\top}(\tbthetaj - \tbthetaz)))\big\|_2 \\
\leq & \big\|(1 - \eta m \lambda)(\bthetaj - \tbthetaj)\big\|_2 + \eta\big\|(\Jb^{(j)} - \Jb^{(0)})(\pb^{(j)} - \yb)\big\|_2 + \eta\big\|\Jb^{(0)}\big\|_2\big\|\pb^{(j)} - \sigma(\Jb^{(0)\top}(\tbthetaj - \tbthetaz))\big\|_2 \\
\leq & \big\|(1 - \eta m \lambda)(\bthetaj - \tbthetaj)\big\|_2 + \eta\big\|(\Jb^{(j)} - \Jb^{(0)})(\pb^{(j)} - \yb)\big\|_2 + k_{\mu}\eta\big\|\Jb^{(0)}\big\|_2\big\|\fb^{(j)} - \Jb^{(0)}{^\top}(\tbthetaj - \tbthetaz)\big\|_2 \\
\leq & \big\|(\Ib - \eta (m \lambda\Ib + k_{\mu}\Hb^{(0)}))\big\|_2\big\|\bthetaj - \tbthetaj \big\|_2 + \eta\big\|\Jb^{(j)} - \Jb^{(0)}\big\|_2\big\|\pb^{(j)} - \yb\big\|_2 \\
& + k_{\mu}\eta\big\|\Jb^{(0)}\big\|_2\big\|\fb^{(j)} - \Jb^{(0)\top}(\bthetaj - \bthetaz)\big\|_2
\end{align*}
where the inequality holds due to triangle inequality, matrix spectral norm inequality, and the Lipschitz continuity of the logistic function.
We now bound the three terms in the RHS separately. 
\begin{align*}
    \big\|\Ib - \eta (m \lambda\Ib + k_{\mu}\Hb^{(0)})\big\|_2 \big\|\bthetaj - \tbthetaj\big\|_2 \leq (1 - \eta m \lambda)\big\|\bthetaj - \tbthetaj\big\|_2,
\end{align*}
where the inequality holds since, $\eta(m\lambda\Ib - k_{\mu}[\Jb^{(0)}{^\top}\Jb^{(0)}]) \preceq \eta(m\lambda\Ib + C_1n_tmL\Ib)  \preceq\Ib$,
for some $C_1 > 0$, where the inequality holds due to the choice of $\eta$. For the second term, we have,
\begin{align*}
    \eta\big\|\Jb^{(j)} - \Jb^{(0)}\big\|_2\big\|\pb^{(j)} - \yb \big\|_2 \leq C_2\eta n_t^P \tau^{1/3}L^{7/2}\sqrt{m \log m},
\end{align*}
for some $C_2 > 0$, where the inequality holds due to Eq~\eqref{eq:jacobcc2} and Lemma~\ref{lemma:linearconvergence}. For the third term,
\begin{align*}
    k_{\mu}\eta\big\|\Jb^{(0)}\big\|_2\big\|\fb^{(j)} - \Jb^{(0)}{^\top}(\bthetaj - \bthetaz)\big\|_2 \leq    C_3\eta n_t^Pm \tau^{4/3}L^{7/2}\sqrt{\log m} 
\end{align*}
for some $C_3 > 0$, where the inequality holds due to Eq~\eqref{eq:jacobcc1} and Eq~\eqref{eq:jacobcc3}. By chaining all the inequalities, we have,
\begin{equation*}
    \big\|\btheta^{j+1} - \tilde \btheta^{j+1}\big\|_2 \leq (1 - \eta m \lambda)\big\|\bthetaj - \tbthetaj\big\|_2 + C_2\eta n_t^P \tau^{1/3}L^{7/2}\sqrt{m \log m} + C_3\eta n_t^Pm \tau^{4/3}L^{7/2}\sqrt{\log m},
\end{equation*}
where $C_4 > 0 $ is a constant. By recursively applying the above inequality from 0 to $j$, we have,
\begin{align*}
    \big\|\btheta^{j+1} - \tilde \btheta^{j+1}\big\|_2 \leq& (C_2\eta n_t^P \tau^{1/3}L^{7/2}\sqrt{m \log m} + C_3\eta n_t^Pm \tau^{4/3}L^{7/2}\sqrt{\log m})/({\eta m \lambda}) \leq {\tau}/{2},
\end{align*}
where $C_5 > 0$ is a constant, the equality holds by the definition of $\tau$. The last inequality holds due to the choice of $m$, where $m^{1/6} \geq C_4L^{7/2}(n_t^P)^{2/3}\lambda^{-2/3}\sqrt{\log m}(C_2 + C_3\sqrt{n_t^P/\lambda})$. Therefore, for any $j \in [J]$,  we have
\begin{align*}
    \big\|\bthetaj - \bthetaz\big\|_2 \leq \big\|\tbthetaj - \bthetaz\big\|_2 + \big\|\bthetaj - \tbthetaj\big\|_2 \leq \sqrt{n_t^P / (2m\lambda)} + \tau / 2 = \tau,
\end{align*}
where the first inequality holds due to triangle inequality, the second inequality holds due to Lemma~\ref{lemma:implicitbias}. This inequality also shows the assumption $\big\|\bthetaj - \bthetaz\big\|_2 \leq \tau$ holds for any $j$. Hence, according to Lemma~\ref{lemma:implicitbias}, we have
\begin{align*}
    \big\|\bthetaj - \bthetaz - \hat{\bgamma}_t\big\|_2 \leq (1 - \eta m \lambda)^{j/2} \sqrt{t / (m\lambda)} + C_5 m^{-2/3}L^{7/2}(n_t^P)^{7/6}\lambda^{-7/6}\sqrt{\log m}(1 + \sqrt{n_t^P/\lambda}).
\end{align*}
This completes the proof.
\end{proof}

\subsection{Proof of Lemma~\ref{lemma:glm}}
We first define $\Ab_t = \lambda\Ib + \sum_{s=1}^{t-1} \sum_{(\ip, \jp) \in \Omega^{ind}_s}\gb_{\ijp}^{s,0} \gb_{\ijp}^{s, 0\top} / m$.
By taking the gradient of Eq~\eqref{eq:linear}, we have $\hat \bgamma_t$ as the solution of,
\begin{align*}
    \sum\nolimits_{s=1}^{t-1}\sum\nolimits_{(\ijp) \in \Omega^{ind}_s} \left(\sigma(\la\gb_{\ijp}^{s, 0}, \bgamma \ra) - y_{\ijp}^s\right)\gb_{\ijp}^{s, 0} + m\lambda\bgamma = 0
\end{align*}
Define $q_t(\bgamma) = \sum_{s=1}^{t-1}\sum_{(i\prime, j\prime) \in \Omega^{ind}_s} \sigma(\la\gb_{\ijp}^{s, 0}, \bgamma \ra)\gb_{\ijp}^{s, 0}/m + \lambda\bgamma$ be the invertible function such that the estimated parameter $\hat \bgamma_t$ satisfies $q(\hat \bgamma_t) = y_{\ijp}^s\gb_{\ijp}^{s, 0}/m$.

As logistic function $\sigma(\cdot)$ is continuously differentiable, $\nabla q_t$ is continuous. Hence, according to the Fundamental Theorem of Calculus, we have $q_t(\bgamma^*) - q_t(\hat{\bgamma}_t) = \Qb_t(\bgamma^* - \hat \bgamma_t)$, where $\Qb_t = \int _0^1\nabla q_t \left(l\bgamma^* + (1 - l)\hat\bgamma_t\right) dl$, and $\bgamma^*$ is the optimal solution of Eq~\eqref{eq:linear}, and according to Lemma~\ref{lemma:equal}, $\bgamma^* = \btheta^* - \btheta_0$.

Therefore, $\nabla q_t(\bgamma) = \sum_{s=1}^{t-1}\sum_{(\ip, \jp)\in \Omega^{ind}_s} \dot\sigma(\la\gb_{\ijp}^{s, 0}, \bgamma \ra) \gb_{\ijp}^{s, 0}{\gb_{\ijp}^{s, 0}}^\top / m + \lambda\Ib$, where $\dot{\sigma}$ is the first order derivative of $\sigma(\cdot)$. Accordingly, we have the following inequality,
\begin{align*}
\|\btheta^* - \btheta_0 -\hat \bgamma_t\|_{\bar \Ab_t} =& \|\Qb_t^{-1}(q_t(\btheta^* - \btheta_0) -q_t(\hat \bgamma_t))\|_{\bar \Ab_t} \\
=&\sqrt{(q_t(\btheta^* - \btheta_0) -q_t(\hat \bgamma_t))^\top\Qb_t^{-1}\bar \Ab_t\Qb_t^{-1}(q_t(\btheta^* - \btheta_0) -q_t(\hat \bgamma_t))} \\
\leq& c_{\mu}^{-1}\|q_t(\btheta^* - \btheta_0) -q_t(\hat \bgamma_t)\|_{\bar \Ab_t^{-1}}
\end{align*}
where the first equality is due to the definition of $q_t$ and $\Qb_t$, and the inequality is based on the definition of $c_{\mu}$, which is defined as $c_{\mu} = \inf_{\btheta \in \bTheta} \dot{\sigma}(\bx^\top\btheta)$. It is easy to verify that $c_{\mu} \leq \frac{1}{4}$. Thus, we can conclude that $\Qb_t \succeq c_{\mu} \bar\Ab_t$, which implies that $\Qb_t^{-1} \preceq c_{\mu}^{-1} \bar\Ab_t^{-1}$. 

Based on the definition of $\hat \bgamma_t$ and the assumption on the noisy feedback that $y_{ij}^t = \sigma(h(\xb_i) - h(\xb_j)) + \xi_{ij}^t$, where $\xi_{ij}^t$ is the noise in user feedback, we have 
\begin{align*}
    &\sqrt{m}q_t(\hat \bgamma_t) - \sqrt{m}q_t(\btheta^* - \btheta_0)\\
    =& \sum\nolimits_{s=1}^{t-1}\sum\nolimits_{(\ip, \jp)\in \Omega^{ind}_s}(y_{\ijp} - \la \gb_{\ijp}^{s, 0},~ \btheta^* - \btheta_0\ra)\gb_{\ijp}^{s, 0}/\sqrt{m} - \lambda\sqrt{m}(\btheta^* - \btheta_0)\\
    =& \sum\nolimits_{s=1}^{t-1}\sum\nolimits_{(\ip, \jp)\in \Omega^{ind}_s} \xi_{\ijp}^s\gb_{\ijp}^{s, 0}/\sqrt{m}- \lambda\sqrt{m}(\btheta^* - \btheta_0).
\end{align*}

As $\xi_{\ijp}^s \sim \nu$-sub-Gaussian, according to Theorem 1 in~\citep{abbasi2011improved}, with probability at least $1 - \delta_1$
\begin{align*}
    \|\sqrt{m}(\btheta^* - \btheta_0 -\hat \bgamma_t)\|_{\bar \Ab_t} \leq & c_{\mu}^{-1} \| \sum\nolimits_{s=1}^{t-1}\sum\nolimits_{(\ip, \jp)\in \Omega^{ind}_s} \xi_{\ijp}^s\gb_{\ijp}^{s, 0}/\sqrt{m}- \lambda\sqrt{m}(\btheta^* - \btheta_0)\|_{\bar \Ab_t^{-1}} \\
    \leq &  c_{\mu}^{-1}\Big(\|\sum\nolimits_{s=1}^{t-1}\sum\nolimits_{(\ip, \jp)\in \Omega^{ind}_s} \xi_{\ijp}^s\gb_{\ijp}^{s, 0}/\sqrt{m} \|_{\bar \Ab_t^{-1}} + \sqrt{\lambda m}\|\btheta^* - \btheta_0\|\Big)\\
    \leq & c_{\mu}^{-1}( \sqrt{\nu^2\log {\det (\bar \Ab_t)}/{(\delta_1^2\det (\lambda \Ib))}}+ \sqrt{\lambda}S).
\end{align*}
This completes the proof.

\section{Proofs of lemmas in Appendix C}

\subsection{Proof of Lemma~\ref{lemma:newboundregret}}
As defined before, we assume that there are in $n_T$ possible document candidate to be evaluated during the model learning, and there are $N_T^P$ possible document pairs to be evaluated, and $N_T^P = n_T^2/2$. 
Then, we have the following quantities.
\begin{align*}
    &\Gb = [\gb(\xb_1; \btheta_0)/\sqrt{m},\dots,\gb(\xb_{n_T}; \btheta_0)/\sqrt{m}] \in \RR^{p \times n_T} \\
    &\hat \Gb = [\Gb, (\gb(\xb_{1, 1}; \btheta_0) - \gb(\xb_{1, 2}; \btheta_0))/\sqrt{m},\dots, (\gb(\xb_{N, 1}; \btheta_0) - \gb(\xb_{N_T^P, 2}; \btheta_0))/\sqrt{m}] \in \RR^{p \times (N_T^P + n_T)}
\end{align*}

Based on the $\Hb$ defined in Definition~\ref{def:ntk}, of which the effective dimension of $\hat \Hb$ is defined as,
\begin{align}
    \tilde d_N = (\log \det (\Ib + \hat \Hb/\lambda))/(\log (1+(N_T^P + n_T)/\lambda)) \label{eq:effectivedn}.
\end{align}

\begin{proof}[Proof of Lemma~\ref{lemma:newboundregret}]
According to Lemma 11 in \citet{abbasi2011improved}]\label{lemma:oriself}, we have the following inequality:
\begin{align}
    \sum\nolimits_{t=1}^T\sum\nolimits_{(\ip, \jp) \in \Omega_t} \min\bigg\{\|\gb_{\ijp}^{t, 0}/\sqrt{m}\|_{\Ab_t^{-1}}^2,1\bigg\} \leq 2 \log \frac{\det \Ab_T}{\det \lambda \Ib}.\notag
\end{align}
Based on the definition of $\Gb$ and $\hat \Gb$, we have,
\begin{align*}
    \log \frac{\det \Ab_T}{\det \lambda \Ib}
    & = \log \det \bigg(\Ib + \sum\nolimits_{t=1}^T\sum\nolimits_{(\ip, \jp) \in \Omega_t^{ind}} \gb_{\ijp}^{t, 0}\gb_{\ijp}^{t, 0\top}/(m\lambda) \bigg)\notag \\
    & \leq \log \det \bigg(\Ib + \sum\nolimits_{i=1}^{N} \gb_i\gb_i^\top/(m\lambda) \bigg)\notag  = \log \det \bigg(\Ib + \hat \Gb \hat \Gb^\top/\lambda \bigg)\notag  = \log \det \bigg(\Ib + \hat \Gb^\top\hat \Gb/\lambda \bigg),\label{selfnormal:2}
\end{align*}
where the inequality holds naively, the third equality holds since for any matrix $\Mb \in \RR^{p \times N}$, we have $\det (\Ib + \Mb\Mb^\top) = \det (\Ib + \Mb^\top\Mb)$.  
Therefore, we have
\begin{align*}
     \log \det \bigg(\Ib + \hat \Gb^\top\hat \Gb/\lambda \bigg) 
     =& \log \det \bigg(\Ib + \hat \Hb/\lambda + (\hat\Gb^\top\hat\Gb - \Hb)/\lambda \bigg)\\ 
     \leq&  \log \det \bigg(\Ib + \hat \Hb / \lambda\bigg) + \la(\Ib + \hat \Hb / \lambda)^{-1}, (\hat \Gb^\top\hat\Gb - \hat \Hb)/\lambda \ra \\
     \leq & \log \det \bigg(\Ib + \hat \Hb / \lambda\bigg) + \|(\Ib + \hat \Hb / \lambda)^{-1}\|_F\|(\hat \Gb^\top\hat\Gb - \hat \Hb)/\lambda\|_F \\
     \leq&  \log \det \bigg(\Ib + \hat \Hb / \lambda\bigg) + \sqrt{N_T^P + n_T}\|\hat \Gb^\top\hat\Gb - \hat \Hb\|_2
\end{align*}
where the first equality stands trivially, the second inequality is due to the convexity of $\log \det (\cdot)$, the third inequality holds due to the fact that $\la \Mb, \Bb\ra \leq \|\Mb\|_F \|\Bb\|_F$, the third inequality holds due to the facts that $\Ib + \hat \Hb/\lambda \succeq \Ib$, $\lambda \geq 1$ and $\|\Mb\|_F \leq \sqrt{N}\|\Mb\|_2$ for any $\Mb \in \RR^{N \times N}$.
According to Lemma~\ref{lemma:ntkconverge}, we know that with properly chosen $m$, $\|\Gb^\top\Gb - \Hb\|_F \leq n\epsilon$. For any $(i, j) \in [N]^2$ in $\hat \Gb^\top\hat\Gb - \hat \Hb$, we have
\begin{align*}
    &|\la\gb(\xb_{i, 1}; \btheta_0) - \gb(\xb_{i, 2};\btheta_0), \gb(\xb_{j, 1}; \btheta_0) - \gb(\xb_{j, 2};\btheta_0)\ra/m - \hat\Hb_{i, j}|\\
    \leq & |\la \gb(\xb_{i, 1}; \btheta_0),  \gb(\xb_{j, 1}; \btheta_0) \ra - \Hb_{(i, 1), (j, 1)}| + |\la \gb(\xb_{i, 1}; \btheta_0),  \gb(\xb_{j, 2}; \btheta_0) \ra- \Hb_{(i, 1), (j, 2)}| \\
    & + |\la \gb(\xb_{i, 2}; \btheta_0),  \gb(\xb_{j, 1}; \btheta_0) \ra- \Hb_{(i, 2), (j ,1)}| + |\la \gb(\xb_{i, 2}; \btheta_0),  \gb(\xb_{j, 2}; \btheta_0) \ra - \Hb_{(i, 2), (j,2)}|.
\end{align*}
Therefore, we have $\|\hat \Gb^\top\hat\Gb - \hat \Hb\|_F \leq 4N_T^P/n_T \|\Gb^\top\Gb - \Hb\|_F$, and by choosing $m$, we have
\begin{align*}
    \log \det \bigg(\Ib + \hat \Hb / \lambda\bigg) + \sqrt{N}\|\hat \Gb^\top\hat\Gb - \hat \Hb\|_2 \leq \tilde d_N\log(1 + (n_T + N_T^P)/\lambda) + 1
\end{align*}
This completes the proof.
\end{proof}

\section{Proofs of lemmas in Appendix \ref{sec:proof4}}

\subsection{Proof of Lemma~\ref{lemma:jacob}}
\begin{proof}[Proof of Lemma~\ref{lemma:jacob}]
With $\tau$ satisfying the condition of Lemmas~\ref{lemma:cao_functionvalue},\ref{lemma:cao_gradientdifference}, and \ref{lemma:cao_boundgradient}, for any $j \in [J]$ at round $t$, we have,
\begin{align}
    \|\Jb^{(j)}\|_F \leq \sqrt{n_t^P} \max_{i \in [n_t^P]} \|\gb(\xb_{i, 1}; \btheta^{(j)}) - \gb(\xb_{i, 2}; \btheta^{(j)})\|_2 \leq \hat 2C_3^z\sqrt{n_t^PmL}
\end{align}
where the first inequality holds due to $\|\Jb^{(j)}\|_F \leq \sqrt{n_t^P}\|\Jb^{(j)}\|_{2, \infty}$, the second inequality holds due to the triangle inequality and Lemma~\ref{lemma:cao_boundgradient}.
Accordingly, we have,
\begin{align}
  \|\Jb^{(j)} - \Jb^{(0)}\|_F \leq& \sqrt{n_t^P}\max_{i \in [n_t^P]}\|\gb(\xb_{i, 1}; \btheta^{(j)}) - \gb(\xb_{i, 2}; \btheta^{(j)}) - (\gb(\xb_{i, 1}; \btheta^{(0)}) - \gb(\xb_{i, 2}; \btheta^{(0)}))\|_2 \nonumber \\
  \leq& \hat 2C_3^wC_3^z \sqrt{n_t^P m \log m}\tau^{1/3}L^{7/2},
\end{align}
where the second inequality holds due to triangle inequality and lemma~\ref{lemma:cao_boundgradient}.
Similarly, we have,
\begin{align}
    \|\fb^{(s)} -\fb^{(j)} - [\Jb^{(j)}]^\top(\btheta^{(s)} - \btheta^{(j)})\|_F \nonumber
    \leq \hat C_6\tau^{4/3}L^3\sqrt{n_t m \log m},
 \end{align}
 where $C_4 > 0$ and $C_5 > 0$ are constants, the second inequality is based on Lemma~\ref{lemma:cao_functionvalue} with the assumption that $\|\btheta^{(j)} - \btheta^{(0)}\|_2 \leq \tau$ and $\|\btheta^{(s)} - \btheta^{(0)}\|_2 \leq \tau$.

Last, it is easy to have that $\|\yb\|_2 \leq \sqrt{n_t}\max_{i\in[n_t]} |y_i| \leq \sqrt{n_t}$. This completes the proof.
\end{proof}

\subsection{Proof of Lemma~\ref{lemma:linearconvergence}}
\begin{proof}[Proof of Lemma~\ref{lemma:linearconvergence}]
The proof is based on Lemma C.3 in \citep{zhou2019neural}, where the convergence of squared loss is analyzed. In our case, we adopt the cross-entropy loss. In the following, we provide the key difference between our analysis concerning the cross-entropy function.

With $\tau$ satisfying the conditions in Lemmas~\ref{lemma:cao_boundgradient}, \ref{lemma:cao_functionvalue}, \ref{lemma:cao_gradientdifference}, and the loss function we have for the neural network as,
\begin{align*}
  &\cL_t(\btheta) \\
  =& \sum\nolimits_{s=1}^t\sum\nolimits_{(i, j) \in \Omega_s^{ind}} -(1 - \yijs)\log(1 - \sigma(f_{ij})) -  \yijs\log(\sigma(f_{ij})) + \frac{m \lambda}{2}\|\btheta - \btheta_0\|^2 \\
  = & \sum\nolimits_{i=1}^{n_t} -y_i\log(\sigma(f(\xb_{i, 1}; \btheta) - f(\xb_{i, 2}; \btheta)) - (1 - y_i) \log(1 - \sigma(f(\xb_{i, 1}; \btheta) - f(\xb_{i, 2}; \btheta)) + \frac{m\lambda}{2}\|\btheta - \btheta_0\|^2
\end{align*}
with the first equation the same as Eq~\eqref{eq:loss} in Section \ref{sec:method}, and we re-write the loss with $n_t$, which is defined in the proof lemma~\ref{lemma:newboundreference}.

We need the following quantities,
\begin{align*}
    &\Jb(\btheta) = \Big(\gb(\xb_{1, 1}; \btheta) - \gb(\xb_{1, 2}; \btheta), \dots, \gb(\xb_{n_t, 1}; \btheta) - \gb(\xb_{n_t, 2}; \btheta)\Big) \in \RR^{p \times n_t} \\ 
    & \fb(\btheta) = \Big(f(\xb_{1, 1}; \btheta) - f(\xb_{1, 2}; \btheta), \dots, f(\xb_{n_t, 1}; \btheta) - f(\xb_{n_t, 2}; \btheta)\Big)^\top \in \RR^{n_t \times 1} \\
    & \pb(\btheta) = \Big(\sigma(f(\xb_{1, 1}; \btheta) - f(\xb_{1, 2}; \btheta)), \dots, \sigma(f(\xb_{n_t, 1}; \btheta) - f(\xb_{n_t, 2}; \btheta))\Big)^\top \in \RR^{n_t \times 1} \\
    &\yb = \big(y_{1}, \dots y_{n_t}\big)^\top \in \RR^{n_t \times 1}.
\end{align*}
First, the cross entropy loss, $l = \sum_{k=1}^K -y_k\log(\sigma(s_k)) - (1-y_k)\log(1 - \sigma(s_k))$ is convex and $\frac{1}{4}$-smooth. The convexity is trivial to prove.  For the smoothness, we have $\frac{\partial l}{\partial s_k} = \sigma(s_k) - y_k$, and $\frac{\partial^2 l}{\partial s_k^2} = \sigma(s_k)(1 - \sigma(s_k))$. As $\sigma(s_k) \in [0, 1]$, $\frac{\partial^2 l}{\partial \bs ^2} \preceq \frac{1}{4}\Ib$.

Based on the smoothness of cross entropy loss function, we have for arbitrary $\btheta$ and $\btheta^\prime$ 
\begin{align}
    &\cL_t(\btheta^\prime) - \cL_t(\btheta) \\
    \leq & \la  \pb(\btheta) - \yb, \fb(\btheta^\prime - \fb(\btheta))\ra + \frac{1}{8}\|\fb(\btheta^\prime) - \fb(\btheta)\|^2 + m \lambda \la\btheta - \btheta^0, \btheta^\prime - \btheta  \ra+ \frac{m\lambda}{2}\|\btheta^\prime - \btheta\|^2 \nonumber\\ 
    \leq & \la\nabla \cL(\btheta), \btheta^\prime - \btheta \ra + \|\pb(\btheta) - \yb\|_2\|\eb\|_2 + \frac{1}{8}\|\eb\|^2  + \frac{C_e}{8}(m\lambda + n_tmL)\|\btheta^\prime - \btheta\|^2 \label{eq:6}
\end{align}
where $\eb = \fb(\btheta^\prime) - \fb(\btheta) - \Jb(\btheta)^\top(\btheta^\prime - \btheta)$, the last inequality is based on Lemma~\ref{lemma:jacob}.
By the convexity of cross entropy loss, we have,
\begin{align}
    \cL_t(\btheta^\prime) - \cL_t(\btheta) \geq& \la\pb(\btheta) - \yb, \fb(\btheta^\prime) - \fb(\btheta)\ra + m \lambda \la\btheta - \btheta^0, \btheta^\prime - \btheta  \ra+ \frac{m\lambda}{2}\|\btheta^\prime - \btheta\|^2 \notag \\
    \geq&  -\frac{\|\nabla \cL(\btheta)\|^2}{2m\lambda} - \|\pb(\btheta) - \yb\|_2\|\eb\|_2 \label{eq:7},
\end{align}
where the third inequality is based on Cauchy-Schwarz inequality, the fourth inequality is based on the fact that $\la\ab, \xb + c\|\xb\|_2^2 \ra \geq -\|\ab\|_2^2/(4c) $ for any vectors $\ab$, $\xb$ and. $c > 0$.

Taking $\btheta^\prime = \btheta - \eta \nabla\cL(\btheta)$ for Eq \eqref{eq:6} and substituting Eq~\eqref{eq:7} into Eq~\eqref{eq:6}, we have
\begin{align*}
    \cL_t(\btheta - \eta \nabla\cL(\btheta)) - \cL_t(\btheta) \leq& -\eta(1 -\frac{C_e}{8}(m\lambda + n_tmL)\eta)\|\nabla\cL(\btheta)\|^2 + \|\pb(\btheta) - \yb\|_2\|\eb\|_2 + \frac{1}{8}\|\eb\|^2 \\
    \leq& m\lambda\eta (\cL(\btheta^\prime) - \cL(\btheta) /2) + \|\eb\|^2_2 (1 + 2m\lambda\eta + 2/(m\lambda\eta)).
\end{align*}
Interested readers can refer to \citep{zhou2019neural} for the details of the derivations. It is easy to verify that $\|\pb(\btheta) - \yb\|^2 \leq 2\cL(\btheta)$. Therefore, by taking $\btheta = \btheta^{(j)}$ and $\btheta^{\prime} = \btheta^{(0)}$, we have
\begin{align*}
    \cL(\btheta^{(j+1)}) - \cL(\btheta^{(0)})\notag \leq  (1-m\lambda\eta/2)\big[\cL(\btheta^{(j)}) - \cL(\btheta^{(0)})\big] + m\lambda\eta/2\cL(\btheta^{(0)}) + \|\eb\|_2^2\big(1 + 2 m \lambda\eta + 2/(m \lambda\eta)\big)\notag
\end{align*}
We have $\cL(\btheta^{(0)}) = n_t \log 2 \leq n_t$, and $\|\eb\|_2^2\big(1 + 2 m \lambda\eta + 2/(m/\lambda\eta)\big)\leq m\lambda\eta n_t / 2$ from \citep{zhou2019neural}, we have $\cL(\btheta^{(j+1)}) - \cL(\btheta^{0}) \leq 2n_t, \|\pb^{(j+1)} - \yb\|_2 \leq 2 \sqrt{n_t}$

This completes the proof.
\end{proof}

\subsection{Proof of Lemma~\ref{lemma:implicitbias}}
\begin{proof}[Proof of Lemma \ref{lemma:implicitbias}]
It can be verified that $\tau$ satisfies the conditions of Lemma \ref{lemma:jacob}, thus Lemma \ref{lemma:jacob} holds. We know that $\tilde\btheta^{(j)}$ is the sequence generated by applying gradient descent on the following problem:
\begin{align*}
    \min_{\btheta}\tilde \cL(\btheta) =& \sum\nolimits_{i=1}^{n_t} -(1 - y_i)\log(1 - \sigma((\gb(\xb_{i, 1}; \btheta^{(0)}) - \gb(\xb_{i, 2}; \btheta^{(0)}))^\top(\btheta - \btheta^{(0)}))) \\
    & -  y_i\log(\sigma((\gb(\xb_{i, 1}; \btheta^{(0)}) - \gb(\xb_{i, 2}; \btheta^{(0)}))^\top(\btheta - \btheta^{(0)}))) + \frac{m \lambda}{2}\|\btheta - \btheta_0\|^2 
\end{align*}
Therefore $\|\btheta^{(0)} - \tilde\btheta^{(j)}\|_2$ can be bounded as
\begin{align*}
    \frac{m\lambda}{2}\|\btheta^{(0)} - \tilde\btheta^{(j)}\|_2^2 \leq & \sum\nolimits_{i=1}^{n_t} -(1 - y_i)\log(1 - \sigma((\gb(\xb_{i, 1}; \btheta^{(0)}) - \gb(\xb_{i, 2}; \btheta^{(0)}))^\top(\tilde \btheta^{(j)} - \btheta^{(0)}))) \\
    & -  y_i\log(\sigma((\gb(\xb_{i, 1}; \btheta^{(0)}) - \gb(\xb_{i, 2}; \btheta^{(0)}))^\top(\tilde \btheta^{(j)} - \btheta^{(0)}))) + \frac{m \lambda}{2}\|\tilde \btheta^{(j)} - \btheta_0\|^2  \\
    \leq & \sum\nolimits_{i=1}^{n_t} -(1 - y_i)\log(1 - \sigma((\gb(\xb_{i, 1}; \btheta^{(0)}) - \gb(\xb_{i, 2}; \btheta^{(0)}))^\top(\tilde \btheta^{(0)} - \btheta^{(0)}))) \\
    & -  y_i\log(\sigma((\gb(\xb_{i, 1}; \btheta^{(0)}) - \gb(\xb_{i, 2}; \btheta^{(0)}))^\top(\tilde \btheta^{(0)} - \btheta^{(0)}))) + \frac{m \lambda}{2}\|\tilde \btheta^{(0)} - \btheta_0\|^2.
\end{align*}
It is easy to verify that $\tilde \cL$ is a $m\lambda$-strongly convex function and $C_1(n_ tmL + m\lambda)$-smooth function for some positive constant $C_1$, since 
\begin{align}
    \nabla^2 \tilde \cL \preceq \big(\frac{1}{4}\big\|\Jb^{(0)}\big\|_2^2 + m\lambda\big)\Ib \preceq C_1(n_tmL + m\lambda),\notag
\end{align}
where the first inequality holds due to the definition of $\tilde \cL$, the second inequality holds due to Lemma \ref{lemma:jacob}. Since we choose $\eta \leq C_1(n_tmL + m\lambda)^{-1}$, then by standard results of gradient descent on ridge linear regression, $\tilde\btheta^{(j)}$ converges to $\btheta^{(0)} + \hat{\bgamma}_t$ with a convergence rate specified as follows,
\begin{align*}
    \big\|\tilde\btheta^{(j)} - \btheta^{(0)} - \hat \bgamma_t\big\|_2^2 \leq& (1 - \eta m \lambda)^j\cdot \frac{2}{m\lambda}(\cL(\btheta^{(0)}) - \cL\big(\btheta^{(0)} + \hat \bgamma_t \big))\notag \\
    \leq& \frac{2(1 - \eta m \lambda)^j}{m\lambda}\cL(\btheta^{(0)}) \leq (1 - \eta m \lambda)^jn_t,
\end{align*}  
where the first inequality holds due to the convergence result for gradient descent and the fact that $\btheta^{(0)} + \hat \bgamma_t$ is the minimal solution to $\cL$, the second inequality holds since $\cL \geq 0$, the last inequality holds due to Lemma \ref{lemma:jacob}.
\end{proof}

\end{document}